\documentclass[a4paper, USenglish, cleveref, autoref, nameinlink, numberwithinsect, thm-restate]{lipics-v2021}              

\nolinenumbers
\hideLIPIcs

\usepackage{dsfont}
\usepackage{todonotes}
\usepackage[ruled,vlined,linesnumbered]{algorithm2e}
\usepackage{mathtools}

\hypersetup{colorlinks = true, citecolor=blue}

\definecolor{skyblue}{rgb}{0.53, 0.81, 0.92}
\definecolor{vermilion}{rgb}{0.89, 0.26, 0.2}

\newcommand{\defproblem}[3]{
  \vspace{3mm}
\noindent\fbox{
  \begin{minipage}{.95\textwidth}
  \begin{tabular*}{\textwidth}{@{\extracolsep{\fill}}lr} #1  \\ \end{tabular*}
  {\bf{Input:}} #2  \\
  {\bf{Goal:}} #3
  \end{minipage}
  }
  \vspace{2mm}
  }

\newcommand{\defdecproblem}[3]{
  \vspace{3mm}
\noindent\fbox{
  \begin{minipage}{.95\textwidth}
  \begin{tabular*}{\textwidth}{@{\extracolsep{\fill}}lr} #1  \\ \end{tabular*}
  {\bf{Input:}} #2  \\
  {\bf{Question:}} #3
  \end{minipage}
  }
  \vspace{2mm}
  }

\newcommand{\OO}{\mathcal{O}}
\newcommand{\Oh}{\mathcal{O}}

\newcommand{\SSS}{\ensuremath{\mathcal{S}}\xspace}

\newcommand{\IBEs}{(\varphi_i)_{i \in [1,b]}}

\newcommand{\LMC}[1][$c$]{\textsc{LS Max #1-Cut}\xspace}
\newcommand{\MC}[1][$c$]{\textsc{Max #1-Cut}\xspace}

\newcommand{\W}[1][1]{\ensuremath{\mathrm{W}[#1]}}

\newcommand{\nd}{\mathrm{nd}}

\newcommand{\colo}{\chi}

\newcommand{\val}{{\rm val}}
\newcommand{\score}{{\rm score}}
\newcommand{\faults}{{\rm faults}}
\newcommand{\Nash}{{\rm Nash}}

\newcommand{\comp}{Schmitz Cargobull AG}

\newcommand{\nn}{{\mathbb N}}

\newcommand{\vect}{\textsf{vec}}
\newcommand{\bvect}{\textsf{B-vec}}
\newcommand{\di}{\textsf{d}}
\newcommand{\dy}{\textsf{dy}}
\newcommand{\db}{\textsf{db}}

\newcommand{\oh}{\textsf{ol}}

\newcommand{\Bin}{\textsc{Vector Bin Packing}\xspace}
\newcommand{\BinLoc}{\textsc{LS Vector Bin Packing}\xspace}

\newcommand{\NSW}{\textsc{Nash Social Welfare}\xspace}
\newcommand{\NSWLoc}{\textsc{LS Nash Social Welfare}\xspace}

\newcommand{\LVBP}{\BinLoc}

\newcommand{\MK}{\textsc{Multi Knapsack}\xspace}
\newcommand{\MKLoc}{\textsc{LS Multi Knapsack}\xspace}

\newcommand{\VC}{\textsc{Vertex Cover}\xspace}
\newcommand{\LVC}{\textsc{LS Vertex Cover}\xspace}

\newcommand{\CE}{\textsc{Cluster Editing}\xspace}
\newcommand{\LSCE}{\textsc{LS Cluster Editing}\xspace}

\newcommand{\GBP}{\textsc{Generalized Bin Problem}\xspace}
\newcommand{\LSGBP}{\textsc{LS \GBP}\xspace}

\newcommand{\MCPi}{\textsc{Multi Component~$\Pi$ Deletion}\xspace}
\newcommand{\LSMCPi}{\textsc{LS Multi Component~$\Pi$ Deletion}\xspace}

\newcommand{\PiDel}{\textsc{$\Pi$ Vertex Deletion}\xspace}
\newcommand{\LSPiDel}{\textsc{LS $\Pi$ Vertex Deletion}\xspace}

\newcommand{\sv}[1]{}

\usepackage{tikz}
\usetikzlibrary{decorations.shapes,decorations.pathreplacing,decorations.pathmorphing}
\usetikzlibrary{arrows,matrix,shapes}
\usetikzlibrary{positioning}
\usetikzlibrary{patterns}
\tikzset{
        stars/.style={star,inner sep=2pt}
    }

\newcommand{\dflip}{d_{\text{flip}}}

\usepackage[ruled,vlined,linesnumbered]{algorithm2e}

\crefname{claim}{Claim}{Claims}

\title{Fantastic Flips and Where to Find Them:\\ A General Framework for Parameterized Local Search on Partitioning Problems\thanks{A preliminary version of this work appeared in the Proceedings of the 19th International Symposium on Algorithms and Data Structures (WADS~'25), Schloss Dagstuhl – Leibniz-Zentrum für Informatik. 
This full version contains all missing proofs.}}
\titlerunning{A General Framework for Parameterized Local Search on Partitioning Problems}

\author{Niels Grüttemeier}{Fraunhofer IOSB-INA, Lemgo, Germany}{niels.gruettemeier@iosb-ina.fraunhofer.de}{https://orcid.org/0000-0002-6789-2918}{Supported by the project \emph{Datenfabrik.NRW}, a project by \emph{KI.NRW}, funded by the Ministry for Economics, Innovation, Digitalization and Energy of the State of North Rhine-Westphalia (MWIDE).}

\author{Nils Morawietz}{Institute of Computer Science, Friedrich Schiller University Jena, Germany\\ LaBRI, Université de Bordeaux, France}{nils.morawietz@uni-jena.de}{https://orcid.org/0000-0002-7283-4982}{Partially supported by the French ANR, project ANR-22-CE48-0001 (TEMPOGRAL).}

\author{Frank Sommer}{Institute of Computer Science, Friedrich Schiller University Jena, Germany}{frank.sommer@uni-jena.de}{https://orcid.org/0000-0003-4034-525X}{Supported by the Alexander von Humboldt Foundation.}

\authorrunning{N. Grüttemeier, N. Morawietz, F. Sommer}
\Copyright{Niels Grüttemeier, Nils Morawietz, Frank Sommer}



\relatedversion{}

\supplement{}

\acknowledgements{We would like to thank Lea Hoeing, Kai Lebbing, Christiane Markuse-Schneider, and Lukas Ptock (Schmitz Cargobull AG) for showing us the interesting application of \Bin.}

\begin{document}
\maketitle

\begin{abstract}
Parameterized local search combines classic local search heuristics with the paradigm of parameterized algorithmics. 
While most local search algorithms  aim to improve given solutions by performing one single operation on a given solution, the parameterized approach aims to improve a solution by performing~$k$ simultaneous operations. Herein,~$k$ is a parameter called \emph{search radius} for which the value can be chosen by a user. One major goal in the field of parameterized local search is to outline the trade-off between the size of~$k$ and the running time of the local search step.

In this work, we introduce an abstract framework that generalizes natural parameterized local search approaches for a large class of partitioning problems: 
Given~$n$ items that are partitioned into~$b$ bins and a target function that evaluates the quality of the current partition, one asks whether it is possible to improve the solution by removing up to~$k$ items from their current bins and reassigning them to other bins. 
Among others, our framework applies for the local search versions of problems like \textsc{Cluster Editing}, \textsc{Vector Bin Packing}, and \textsc{Nash Social Welfare}.
Motivated by a real-world application of the problem \textsc{Vector Bin Packing}, we introduce a parameter called \emph{number of types}~$\tau \leq n$ and show that all problems fitting in our framework can be solved in $\tau^k \cdot 2^{\Oh(k)} \cdot |I|^{\Oh(1)}$~time, where~$|I|$ denotes the total input size. 
In case of~\textsc{Cluster Editing}, the parameter~$\tau$ is the well-known parameter neighborhood diversity of the input graph.

We complement these algorithms by showing that for all considered problems, an algorithm with running time~$\tau^{o(k)} \cdot 2^{\Oh(k)} \cdot |I|^{\Oh(1)}$ does not exist unless the Exponential Time Hypothesis fails. 
Additionally, we show that even on very restricted instances, all considered problems are W[1]-hard when parameterized by the search radius~$k$ alone.
In case of the local search version of~\textsc{Vector Bin Packing}, we provide an even stronger W[1]-hardness result.
\end{abstract}

\section{Introduction}
The principle of local search is among the most important heuristic approaches in combinatorial optimization and it is highly relevant to find good solutions to NP-hard problems in practice~\cite{HS04}. 
The idea is to apply small modifications on a given starting solution to obtain a new solution with a better target value than the starting solution. Local search has been studied extensively and it has been proven to be highly efficient~\cite{CSLS13,HS04,LHC+20,DBLP:conf/iwpec/BlasiusFGHHSWW21a}. Furthermore, it is easy to understand and it is also a good plugin to improve already competitive solutions provided by other metaheuristics~\cite{OX14,GGKM23}.

Consider a classic partitioning problem as \MK, where the goal is to assign items (each with a weight) to multiple knapsacks (each with a weight capacity and specific values for the items) in a way that all capacity constraints are satisfied and the total value is maximal. One of the most natural ways to apply small modifications in a local search scenario is an \emph{item flip}, where one removes a single item from its current knapsack and inserts it into another knapsack. 
This natural idea of flipping the assignment of single items to improve a solution is studied for knapsack problems~\cite{DF01} and for other partitioning problems~\cite{FPRR02,LMV99}. A general drawback in performing these single item flips---and also in local search in general---is the chance of getting stuck in poor local optimal solutions. Thus, to obtain a robust local search application, a strategy to prevent getting stuck in these poor local optima is required.

In this work, we consider the approach of \emph{parameterized local search}~\cite{M08,FFLRSV12} to decrease the chance of getting stuck in poor local optima reached by item flips. Instead of searching for an improving solution by \emph{one} single operation (here: a single item flip), the user may set a search radius~$k$ to extend the search space for possible improvements that can be reached with up to~$k$ simultaneous operations. 
Thus, the idea is to make the search space larger so that getting stuck in poor local optima becomes less likely. 
Parameterized local search combines local search with the paradigm of parameterized algorithmics~\cite{C+15} and aims to outline the trade-off between the size of the search radius~$k$ and the running time of the search step. 
Parameterized local search has been studied extensively in the algorithmic community: for vertex deletion and partitioning problem in graphs~\cite{FFLRSV12,GKO+12,GHNS13,HN13,DGKW14,KK17,GMNW23,GGKM23}, for problems on strings and phylogenetic networks~\cite{GHK14,KLMS23}, and many other problems~\cite{M08,S11,GKO21,GKM21,HSS25}. 
One major goal in parameterized local search is to show that finding an improvement within search radius~$k$ is fixed parameter tractable (FPT) for~$k$. 
That is, finding an improving solution can be done in time~$g(k) \cdot |I|^{\Oh(1)}$, where~$|I|$ is the total encoding length of the input instance and~$g$ is some computational function only depending on~$k$. 
Note that an algorithm with such a running time nicely outlines the trade-off between radius size and running time, as the superpolynomial part only depends on~$k$ while~$|I|$ only contributes as a polynomial factor to the running time. 
Unfortunately, most parameterized local search problems are W[1]-hard for~$k$~\cite{M08,FFLRSV12,GHK14,M24} and therefore, an algorithm with running time~$g(k) \cdot |I|^{\Oh(1)}$ presumably does not exist. 

Motivated by the negative results for the parameter~$k$, one often studies the combination of~$k$ and some structural parameter~$\tau$ to obtain algorithms with running~$g(k,\tau) \cdot |I|^{\Oh(1)}$. 
This approach has been successful both from a theoretical and experimental perspective~\cite{HN13,KK17,GGJMR19,GGKM23}. In this work, we follow this direction by studying parameterization by~$k$ and an additional parameter~$\tau$ that we call the \emph{number of types}. 
We consider the local search versions of a large class of well-known combinatorial problems including \MC, \MK, \CE and~\Bin, where the search space is defined by performing~$k$ flips. 
Recall that one flip intuitively removes one item from its assigned set and inserts in into some other set. 
The interpretation of flips and of our parameter~$\tau$ always depends on the concrete problem and will be explained for each problem individually.

Our approach of exploiting the parameter~$\tau$ is motivated by a real-world production planning application at the company~\comp{}. 
Planning the production at~\comp{} corresponds to solving an instance of~\Bin~\cite{M23}. 
In~\Bin, one is given a large collection of vectors~$\mathcal{S} \subseteq {\mathds{N}}^d$ together with a vector~$w \in \mathds{N}^d$ and an integer~$b$. 
The question is, whether there exists a partition of~$\mathcal{S}$ into~$b$ parts~$S_1, \dots, S_b$, such that~$\sum_{v \in S_i} v \leq w$ for all~$i \in [1,b]$. 
In the application at~\comp, we have~$\mathcal{S} \subseteq \{0,1\}^d$ and the vectors~$v \in \mathcal{S}$ correspond to customer orders where the entries of~$v$ specify whether a specific option is chosen~($\widehat{=}$ the entry has value~$1$) or not~($\widehat{=}$ the entry has value~$0$). 
The entries of~$w$ correspond to the production capacities available for the corresponding product option at one day of production. 
Instead of asking for a partition where each resulting vector set satisfies~$\sum_{v \in S_i} v \leq w$, 
one asks for a partition minimizing the \emph{total overload}. 
More precisely, the \emph{overload} of a single set~$S_i$ is defined as~$\sum_{j=1}^d \max{(0,(\sum_{v \in S_i} v_j) - w_j)}$, and the \emph{total overload} is the sum of all overloads of all sets~$S_1, \dots, S_b$. 
Note that the total overload is~$0$ if and only if the given instance is a yes-instance of the decision version of~\Bin. 
While \comp{} usually receives a large number of customer orders,
 relatively many of these orders request the exact same product option combinations. 
Consequently, the number of distinct vectors in~$\mathcal{S}$ is much smaller than~$|\mathcal{S}|$. 
Therefore, our research is motivated by setting~$\tau$ to be the number of distinct vectors in an instance of \Bin and study parameterized local search for the combination of~$k$ and~$\tau$ where the target is to minimize the total overload.

From a more abstract point of view, \Bin is a problem where one assigns a collection of \emph{items} (vectors) to a collection of \emph{bins} (sets of the resulting partition) in a way that a target function (total overload) is minimized. 
Furthermore, if two elements from~$\mathcal{S}$ have the same vector, these elements have the same effect on the target function. 
In other words, there are only~$\tau$ distinct ways in which a specific item might have an influence on the target function when assigned to a specific bin. 
This more abstract view leads to a framework providing running times~$\tau^k \cdot 2^{\Oh(k)} \cdot |I|^{\Oh(1)}$ and~$k^{\Oh(\tau)} \cdot |I|^{\Oh(1)}$ for the parameterized local search versions for a wide range of well-known combinatorial problems that behave in the same way as \Bin. 
The $\tau^k \cdot 2^{\Oh(k)} \cdot |I|^{\Oh(1)}$~running-time is particularly motivated since the value of~$k$ is a small constant chosen by the user~\cite{KM22}.

Recall that the concrete interpretation of the parameter~$\tau$, of the items, and of the bins always depends on the concrete problem and will be explained for each problem individually. 
Besides the number of distinct vectors, $\tau$ can be interpreted as the neighborhood diversity of an input graph in case of \MC or \CE, or as the number of distinct value-weight combinations in case of \MK.

\subparagraph{Our Contributions.} In the first part (see \Cref{Section: Framework}), we introduce the \LSGBP, which is a general framework capturing many parameterized local search problems where the search radius is defined by a number of item flips. 
Moreover, we introduce the parameter \emph{number of types} $\tau$ as an abstract concept for \LSGBP capturing well-known parameters such as the neighborhood diversity of a graph. We describe general algorithms for the abstract \LSGBP leading to running times of~$\tau^k \cdot 2^{\Oh(k)} \cdot |I|^{\Oh(1)}$ and~$k^{\Oh(\tau)} \cdot |I|^{\Oh(1)}$ for the local search versions of many problems that fit into our framework (see \Cref{Theorem: Framework Algorithms}). 
In our approach, we group items with the same influence on the target function. 
We express flips between bins as simultaneous insertions into some bins and removals from other bins and calculate the changes on the target function with respect to the grouping. 
In a dynamic programming algorithm, we draw up the balance between insertions and removals, such that the resulting changes on the solution correspond to at most~$k$ item flips.

\begin{table} [t]
\caption{The corresponding local search problems for the flip distance of these problems can be solved in $\tau^k \cdot 2^{\Oh(k)} \cdot |I|^{\Oh(1)}$~time and in $k^{\Oh(\tau)} \cdot |I|^{\Oh(1)}$~time for the respective parameter~$\tau$ as we show in \Cref{Section: Framework Application}. 
In case of \VC, we have exactly two bins: the resulting vertex cover (vc) and the remaining independent set (is). 
In this case, the flip distance corresponds to the cardinality of the symmetric difference between the current vertex cover and the improving vertex cover.}
\centering
\begin{tabular}{lllll}
Problem & Items & Bins & Parameter~$\tau$ & Section\\
\hline
\MC & vertices & color classes & neighborhood diversity & \ref{sec:MC} \\
\NSW & items & agents & number of distinct items & \ref{sec:NSW} \\
\CE & vertices & clusters  & neighborhood diversity & \ref{sec:CE}\\
\Bin & vectors & bins & number of distinct vectors & \ref{sec:VBP} \\
\MK & items & knapsacks  & number of distinct items & \ref{sec:MK} \\
\VC & vertices & (vc, is) & neighborhood diversity & \ref{sec:VDD}
\end{tabular}
\label{Table: Result Overview}
\end{table}

In the second part (see \Cref{Section: Framework Application}) we provide simple example applications of the introduced framework leading to new results for the parameterized local search versions of classic combinatorial optimization problems, graph problems, and one problem from computational social choice (\NSW). 
The formal problem definitions can be found in \Cref{Section: Framework Application}. 
All results in this part mainly rely on simply reformulating the concrete problem as \LSGBP. 
An overview of the studied problems is given in \Cref{Table: Result Overview}. To the best of our knowledge, this is the first work studying parameterized local search for \NSW, \Bin, and \MK. 

In the third part (see \Cref{sec:hardness}), we complement our results by studying parameterization by the search radius~$k$ alone.  
We show that the parameterized local search versions of \NSW and \MK are W[1]-hard when parameterized by~$k$. 
Furthermore, we provide a strong hardness result for \Bin showing W[1]-hardness for~$k+q$, where~$q$ denotes the maximum number of non-zero-entries over all vectors from the input. 
The parameter~$q$ is particularly motivated by the real-world application from the company~\comp, where~$q$ is even smaller than~$\tau$.
We finally show that all of our algorithms with running time~$\tau^k \cdot 2^{\Oh(k)} \cdot |I|^{\Oh(1)}$ are tight in the sense of the \emph{Exponential Time Hypothesis (ETH)}~\cite{IPZ01}.

\section{Preliminaries}
For details on parameterized complexity, we refer to the standard monographs~\cite{C+15,DF13}.

For integers~$a$ and~$b$ with~$a \le b$, we define~$[a, b] := \{i \in \mathds{N} \mid a \le i \le b\}$.
Given a set~$X$ and some integer~$b$, we call a mapping~$f: X \rightarrow [1,b]$ a~\emph{$b$-partition of~$X$} or a~\emph{$b$-coloring of~$X$}. For every~$i \in [1,b]$, we let~$f^{-1}(i):= \{ x \in X \mid f(x) = i\}$. The \emph{flip} between two~$b$-partitions~$f$ and~$f'$ is defined as~$D_{\rm flip}(f,f') := \{ x \in X \mid f(x) \neq f'(x)\}$. The \emph{flip distance} between~$f$ and~$f'$ is then defined as~$d_{\rm flip}(f,f') := |D_{\rm flip}(f,f')|$.
For a set~$X$, we let~$2^X$ denote the \emph{power set} of~$X$.

For a graph~$G=(V,E)$, by~$n:=|V|$ we denote the \emph{number of vertices} and by~$m:= |E|$ we denote the \emph{number of edges}.
By~$N(u):=\{w\in V \mid \{u,w\}\in E\}$ and by~$N(S):=(\bigcup_{u\in S}N(u))\setminus S$ we denote the \emph{open neighborhood} of~$u$ and~$S$, respectively.
Two vertices~$u$ and~$w$ have the \emph{same neighborhood class} if~$N(u)\setminus \{w\} = N(w)\setminus \{u\}$.
The \emph{neighborhood diversity}~$\nd(G)$ is the number of neighborhood classes of~$G$.
A vertex set~$S$ is a \emph{vertex cover} for~$G$ if each edge of~$E$ has at least one endpoint in~$S$.
A vertex set~$S$ is an \emph{independent set} of~$G$, if no edge has of~$E$ has both endpoints in~$S$.

\section{Generalized Parameterized Local Search for Partitioning Problems} \label{Section: Framework}

We present a framework that captures many computational problems such as \Bin, \CE, and \NSW. On a high level, all these problems have in common, that one aims to partition some set~$X$ (e.g. a set of vectors, a set of vertices, or a set of items) into multiple `bins'. A target function assigns a value to the resulting partition. The goal is to find a partition that minimizes (or maximizes) the target function.
This section is structured as follows. We first introduce a `Generalized Bin Problem' that generalizes the computational problems studied in this work. Afterwards, we introduce a parameter called `types'. Finally, we provide the algorithmic results for this parameter.

\subsection{The Generalized Bin Problem}

Intuitively, we aim to partition a given set~$X$ into a given number of `bins'~$b \in \mathds{N}$, and a target value specifies how good this~$b$-partition of~$X$ is. This target value is determined by an `individual bin evaluation', which is a collection of local target values of each single bin. These values are then combined to obtain the target value for the whole~$b$-partition of~$X$.

\begin{definition} \label{Definition: Target Function}
Let~$X$ be a set, let~$b \in \mathds{N}$, and let~$\inf \in \{ \infty, -\infty\}$.
An \emph{individual bin evaluation (IBE)} is a~$b$-tuple~$(\varphi_i)_{i \in [1,b]}$ of functions~$\varphi_i: 2^X \rightarrow \mathds{Z} \cup \{\inf\}$. An IBE defines a \emph{target value}~$\val(f)$ for every~$b$-partition~$f$ by
\begin{align*}
\val(f) := \left( \bigoplus_{i=1}^b \varphi_i(f^{-1}(i)) \right),
\end{align*}
where~$\oplus$ is a commutative and associative binary operation~$\oplus: \mathds{Z} \cup \{ \inf \} \times \mathds{Z} \cup \{ \inf \} \rightarrow \mathds{Z} \cup \{ \inf \}$ satisfying~$a \oplus \inf = \inf \oplus~a = \inf$ for all~$a \in \mathds{Z} \cup \{ \inf \}$.
\end{definition}

While our general framework works for arbitrary commutative and associative operations~$\oplus$, this work only considers concrete problems where~$\oplus$ is either the summation or the multiplication of integer numbers. The value~$\inf  \in \{ \infty, -\infty\}$ corresponds to infeasible assignments of bins, for example, violating capacity constraints of a knapsack. In case of a minimization problem we consider IBE with~$\inf = \infty$ and in case of a maximization problem we have~$\inf= - \infty$. In the remainder of this section, all problem statements and algorithms are given for the case where one aims to minimize the target function. Maximization problems are defined analogously. With \Cref{Definition: Target Function} at hand, we define the following general computational problem for every fixed commutative and associative operation~$\oplus$.

\defproblem{\GBP}
{An integer~$b$, a set~$X$, and an IBE~$(\varphi_i)_{i \in [1,b]}$.}
{Find a~$b$-partition~$f$ that minimizes~$\val(f)$.}

Note that we did not yet specify how the IBE from the input is given. As this depends on the concrete problems, this will be discussed for each problem individually. Throughout this section, we analyze the algorithms running times with respect to the parameter~$\Phi$ denoting the running time needed for one evaluation of a value~$\varphi_i(X')$ with~$X' \subseteq X$.

Recall that our aim is to study parameterized local search for the flip neighborhood. More precisely, we aim to find a~$b$-partition~$f'$ that has a better target value than some given~$b$-partition~$f$, while~$d_{\rm flip}(f,f') \leq k$ for a given~$k$. The corresponding computational problem is defined as follows

\defdecproblem{\LSGBP}
{An integer~$b$, a set~$X$, and an IBE~$(\varphi_i)_{i \in [1,b]}$, a~$b$-partition~$f: X \rightarrow [1,b]$, and an integer~$k$.}{Is there a~$b$-partition~$f'$ with~$d_{\rm flip}(f,f') \leq k$ such that~$\val(f') < \val(f)$?}

\subsection{Types in Generalized Bins}

We next define a parameter called `number of types'~$\tau$. The idea is, that in a polynomial-time preprocessing step, the set~$X$ is partitioned into classes of elements~$(X_1, \dots, X_\tau)$ in a way that all elements in each~$X_i$ have the exact same impact on the target value for every possible bin assignment. The intuitive idea of `same impact' is formalized as follows.

\begin{definition} \label{Definition: Types}
Let~$X$ be a set, let $b \in \mathds{N}$, and let~$(\varphi_i)_{i\in [1,b]}$ be an IBE for~$X$ and~$b$. Two (not necessarily distinct) elements~$x \in X$ and~$y \in X$ are \emph{target equivalent} ($x \sim y$), if for every~$i \in [1,b]$ and for every~$A \subseteq X$ with~$\{x,y\} \cap A = \{x\}$ we have~$\varphi_i ((A\setminus \{x\}) \cup \{y\}) = \varphi_i(A)$.
\end{definition}

\begin{proposition} \label{Proposition: sim is ER}
The relation~$\sim$ is an equivalence relation on~$X$.
\end{proposition}

\begin{proof} 
We obviously have~$x \sim x$ for each~$x \in X$. It remains to show symmetry and transitivity. Throughout this proof, we let~$\varphi$ be an arbitrary function from the IBE.

We first show that~$\sim$ is symmetric. Let~$x \sim y$. Then, for any~$A \subseteq X$ with~$\{x,y\} \cap A = \{y\}$ we have
\begin{align*}
\varphi (A \setminus \{y\} \cup \{x\}) = \varphi ((A \setminus \{y\} \cup \{x\}) \setminus \{x\} \cup \{y\}) = \varphi (A).
\end{align*}
The first equality holds, since~$x \sim y$ and~$\{x,y\} \cap (A \setminus \{y\} \cup \{x\}) = \{x\}$. We thus have~$y \sim x$. Consequently,~$\sim$ is symmetric.

It remains to show transitivity. Let~$x \sim y$ and let~$y \sim z$. We show that~$x \sim z$. To this end, let~$A \subseteq X$ with~$A \cap \{x,z\} = \{x\}$ and consider the following cases.

\textbf{Case 1:} $y \in A$\textbf{.} Then, we have~$A \cap \{y,z\} = \{y\}$. Together with~$y \sim z$, this implies~$\varphi(A) = \varphi(A \setminus \{y\} \cup \{z\})$. Since~$\{x,y\} \cap (A \setminus \{y\} \cup \{z\}) = \{x\}$ and~$x \sim y$, we have~$\varphi(A \setminus \{y\} \cup \{z\})=\varphi( (A \setminus \{y\} \cup \{z\}) \setminus \{x\} \cup \{y\} )$, and thus
\begin{align*}
\varphi(A) = \varphi( (A \setminus \{y\} \cup \{z\}) \setminus \{x\} \cup \{y\} ) = \varphi(A \setminus \{x\} \cup \{z\}).
\end{align*}
Consequently, we have~$x \sim z$.

\textbf{Case 2:} $y \not \in A$\textbf{.} Then, we have~$\{x,y\} \cap A = \{x\}$. Together with~$x \sim y$, this implies~$\varphi(A) = \varphi(A \setminus\{x\} \cup \{y\})$. Since~$\{y,z\} \cap (A \setminus \{x\} \cup \{y\}) = \{y\}$ and~$y \sim z$, we have~$\varphi(A \setminus\{x\} \cup \{y\})=  \varphi((A \setminus\{x\} \cup \{y\}) \setminus \{y\} \cup \{z\})$, and thus
\begin{align*}
\varphi(A) = \varphi((A \setminus\{x\} \cup \{y\}) \setminus \{y\} \cup \{z\}) = \varphi(A \setminus \{x\} \cup \{z\}).
\end{align*}
Consequently, we have~$x \sim z$.
\end{proof}

The main idea of our algorithm is that target equivalent elements can be treated equally as they have the same influence on the target value. When considering concrete problems, we always use a simple pairwise relation between the elements leading to a partition of~$X$ into classes of pairwise target equivalent elements. These classes do not necessarily need to be maximal under this constraint. Thus, it suffices to consider the following relaxation of the equivalence classes.

\begin{definition}
Let~$X$ be a set, let~$b \in \mathds{N}$, and let~$(\varphi_i)_{i\in [1,b]}$ be an IBE for~$X$ and~$b$. A tuple~$(X_1, \dots, X_{\tau})$ of disjoint sets with~$\bigcup_{j=1}^{\tau} X_j = X$ is called a \emph{type partition} of~$X$ if the elements of each~$X_j$ are pairwise target equivalent. For every~$j \in [1, \tau]$, we say that the elements of~$X_j$ \emph{have type~$j$}.
\end{definition}

Obviously, the equivalence classes for~$\sim$ always form a type partition with the minimum number of sets. Throughout this work, we assume that each instance~$I:=(b,X,(\varphi_i)_i, f, k)$ of \LSGBP is associated with a specified type partition, and the parameter \emph{number of types}~$\tau:=\tau(I)$ is defined as the number of sets of the type partition associated with~$I$. 

\subsection{Algorithmic Results for \LSGBP}
Our goal is to study \LSGBP parameterized by~$k+\tau$. 
Applying this on concrete problems then leads to FPT algorithms for a great range of parameterized local search versions of well-known computational problems.
The interpretation of the parameter~$\tau$ always depends on the concrete problem. We provide the following algorithmic results. Recall that~$\Phi$ denotes the running time needed for one evaluation of a value~$\varphi_i(X')$ with~$X' \subseteq X$.

\begin{theorem} \label{Theorem: Framework Algorithms}
\LSGBP can be solved
\begin{enumerate}[$a)$]
\item in $\tau^k \cdot 2^{\Oh(k)} \cdot b \cdot \Phi \cdot |X|^{\Oh(1)}$~time, and
\item in $k^{\Oh(\tau)} \cdot b \cdot \Phi \cdot |X|^{\Oh(1)}$~time
\end{enumerate} 
if a type partition~$(X_1, \dots, X_\tau)$ is additionally given as part of the input.
\end{theorem}

To present the algorithms behind \Cref{Theorem: Framework Algorithms}, we introduce the notion of \emph{type specification} and \emph{type specification operations}. Recall that elements of the same type behave in the same way when assigned to a bin. Thus, when evaluating the target function, one may consider the types of the assigned elements instead of the concrete elements. A type specification is a vector~$\vec{p}$ that intuitively corresponds to a collection of elements in~$X$ containing exactly~$p_j$ elements from the class~$X_j$. 

\begin{definition} \label{Definition: Type Specification}
Let~$b$ be an integer, let~$X$ be a set, and let~$(\varphi_i)_{i\in [1,b]}$ be an IBE. Moreover, let~$(X_1, \dots, X_{\tau})$ be a type partition. A \emph{type specification} is a vector~$\vec{p} = (p_1, \dots, p_\tau) \in {\mathds{N}_0}^\tau$ with~$p_j \in [0,|X_j|]$ for each~$j \in [1,\tau]$.
\end{definition}

Since elements of the same type have the same impact on the target function, we can address the change of target values~$\varphi_i(X')$ by just specifying the types of elements added to and removed from~$X'$ with~$X' \subseteq X$. To ensure that after adding and removing elements of specific types from a set~$X'$ corresponds to an actual subset of~$X$, we introduce the notion of subtractive and additive compatibility.  
\begin{definition} \label{Definition: additive and subtractive compatible}
Let~$b$ be an integer, let~$X$ be a set, and let~$(\varphi_i)_{i\in [1,b]}$ be an IBE. Moreover, let~$(X_1, \dots, X_{\tau})$ be a type partition. We say that a type specification~$\vec{p}$ is
\begin{enumerate}[$a)$]
\item \emph{subtractive compatible} with a set~$X'\subseteq X$ if for every~$j \in [1,\tau]$, we have~$|X_j \cap X'| \geq p_j$.
\item \emph{additive compatible} with a set~$X'\subseteq X$ if for every~$j \in [1,\tau]$, we have~$|X_j \cap (X\setminus X')| \geq p_j$.
\end{enumerate}
For given~$\vec{p}$ and~$\vec{q}$, we use the notation~$(\vec{p},\vec{q}) \propto X'$ if~$\vec{p}$ is subtractive compatible with~$X'$ and~$\vec{q}$ is additive compatible with~$X'$. 
\end{definition}

\begin{definition} \label{Definition: type vector operation}
Let~$b$ be an integer, let~$X$ be a set, and let~$(\varphi_i)_{i\in [1,b]}$ be an IBE. Moreover, let~$(X_1, \dots, X_{\tau})$ be a type partition. For given type specifications~$\vec{p}$ and~$\vec{q}$ with~$(\vec{p},\vec{q}) \propto X'$, the \emph{type vector operation} for each~$i \in [1,b]$ is defined as
\begin{align*}
\varphi_i ((X' \setminus \vec{p}) \cup \vec{q}) := \varphi_i ((X' \setminus X_{\vec{p}}) \cup X_{\vec{q}}),
\end{align*}
where~$X_{\vec{p}} \subseteq X$ is some set containing exactly~$p_j$ arbitrary elements from~$X_j \cap X'$ for every~$j  \in [1,\tau]$, and~$X_{\vec{q}} \subseteq X$ is some set containing exactly~$q_j$ arbitrary elements from~$X_j \cap (X\setminus X')$ for every~$j  \in [1,\tau]$.
\end{definition}

The notion of target equivalence (\Cref{Definition: Types}) together with the notion of subtractive and additive compatibility guarantees the following.

\begin{proposition} \label{Prop: TypeVectorOp Well-def}
The type vector operation is well-defined.
\end{proposition}

\begin{proof}
Since~$\vec{p}$ is subtractive compatible with~$X'$, there are at least~$p_j$ elements in~$X_j \cap X'$ and thus, the set~$X_{\vec{p}} \subseteq X$ exists. Analogously, the set~$X_{\vec{q}}$ exists due to the additive compatibility of~$\vec{q}$ with~$X'$. Consequently,~$(X' \setminus X_{\vec{p}}) \cup X_{\vec{q}} \subseteq X$ and therefore,~$\varphi_i ((X' \setminus X_{\vec{p}}) \cup X_{\vec{q}})$ is defined.

It remains to show that the value~$\varphi_i ((X' \setminus X_{\vec{p}}) \cup X_{\vec{q}})$ does not depend on the choice of elements in~$X_{\vec{p}}$ and~$X_{\vec{q}}$. First, let~$x \in X_{\vec{p}}$ and let~$y \not \in X_{\vec{p}}$ with~$y \sim x$. Then,
\begin{align*}
\varphi_i (X' \setminus (X_{\vec{p}} \setminus \{x\} \cup \{y\}) \cup X_{\vec{q}}) 
= \varphi_i ((X' \setminus X_{\vec{p}} \cup X_{\vec{q}}) \setminus \{y\} \cup \{x\})  = \varphi_i (X' \setminus X_{\vec{p}} \cup X_{\vec{q}}),
\end{align*}
by \Cref{Definition: Types}. Analogously, let~$x \in X_{\vec{q}}$ and let~$y \not \in X_{\vec{q}}$ with~$x \sim y$ we have
\begin{align*}
\varphi_i (X' \setminus X_{\vec{p}} \cup (X_{\vec{q}} \setminus \{x\} \cup \{y\}))
= \varphi_i ((X' \setminus X_{\vec{p}} \cup X_{\vec{q}}) \setminus \{x\} \cup \{y\})
= \varphi_i (X' \setminus X_{\vec{p}} \cup X_{\vec{q}}).
\end{align*} Therefore, the type vector operation is well-defined.
\end{proof}

With the notion of type specifications at hand, we next present the algorithmic results. Intuitively, the algorithms behind \Cref{Theorem: Framework Algorithms} work as follows: One considers every possibility of how many elements of which type belong to the (at most)~$k$ elements in~$D_{\rm flip}(f,f')$. For each such choice, one computes the best improving flip corresponding to the choice. 
Note that the information how many elements of which type are flipped, can be encoded as a type specification~$\vec{\delta}$ in the way that all~$\delta_j$ correspond to the number of flipped elements of type~$j$. We first describe the subroutine computing the improving~$b$-partition~$f'$ corresponding to a given type specification~$\vec{\delta}$. Afterwards, we describe how to efficiently iterate over all possible~$\vec{\delta}$ to obtain the running times stated in \Cref{Theorem: Framework Algorithms}.

We use the notation~$\vec{p} \leq \vec{q}$ to indicate that~$p_i \leq q_i$ for all vector entries, and we let~$\vec{p} - \vec{q}$ denote the component-wise difference between~$\vec{p}$ and~$\vec{q}$.

\begin{lemma} \label{Lemma: DP Algorithm}
Let~$b$ be an integer, let~$X$ be a set, let~$(\varphi_i)_{i \in [1,b]}$ be an IBE, and let~$f:X \rightarrow [1,b]$ be a~$b$-partition. Moreover, let~$(X_1, \dots, X_{\tau})$ be a type partition and let~$\vec{\delta}$ be a type specification and we let~$k:= \sum_{j=1}^{\tau} \delta_j$. A~$b$-partition~$f': X \rightarrow [1,b]$ with
\begin{align*}
(|X_1 \cap D_{\rm flip}(f,f')|, \dots, |X_\tau \cap D_{\rm flip}(f,f')|) \leq \vec{\delta}
\end{align*}
that minimizes~$\val(f')$ among all such~$b$-partitions can be computed in
\begin{enumerate}
\item[$a)$] $2^{\Oh(k)} \cdot b \cdot \Phi \cdot |X|^{\Oh(1)}$~time, or in
\item[$b)$] $(\lceil \frac{k}{\tau} \rceil + 1)^{4\tau} \cdot b \cdot \Phi \cdot |X|^{\Oh(1)}$~time.
\end{enumerate}
\end{lemma}

\begin{proof}
We first provide an algorithm based on dynamic programming and afterwards we discuss the running times~$a)$ and~$b)$.

\textit{Intuition.} Before we formally describe the dynamic programming algorithm, we provide some intuition: Every~$x \in D_{\rm flip}(f,f')$ gets removed from its bin and is inserted into a new bin. We fix an arbitrary ordering of the bins and compute solutions for prefixes of this ordering in a bottom-up manner. For every bin, there is a set of removed elements~$R$ and a set of inserted elements~$I$. Since the target value depends on the types of the elements rather than the concrete sets~$R$ and~$I$, we may only consider how many elements of which type are removed from a bin and inserted into a bin. Therefore, we compute the partial solutions for given~$\vec{p} \leq \vec{\delta}$ and~$\vec{q} \leq \vec{\delta}$, where~$\vec{p}$ corresponds to the removed types, and~$\vec{q}$ corresponds to the inserted types. After the computation, we consider the solution, where~$\vec{p}=\vec{q}=\vec{\delta}$, that is, the solution where the exact same types were removed and inserted. Going back from abstract types to concrete elements then yields the resulting~$b$-partition~$f'$.

\textit{Dynamic programming algorithm.} The dynamic programming table has entries of the form~$T[\vec{p}, \vec{q}, \ell]$ where~$\ell \in [1,b]$,~$\vec{q} \leq \vec{\delta}$, and~$\vec{p} \leq \vec{\delta}$. One such table entry corresponds to the minimal value of 
$$\bigoplus_{i=1}^{\ell} \varphi_i ((f^{-1}(i) \setminus R_i) \cup I_i)$$
under every choice of sets~$R_i \subseteq f^{-1}(i)$ and~$I_i \subseteq X \setminus f^{-1}(i)$ for~$i \in [1,\ell]$ that satisfy
\begin{align*}
&( \sum_{i=1}^{\ell} |R_i \cap X_1|, \dots, \sum_{i=1}^{\ell} |R_i \cap X_{\tau}|) =  \vec{p} \text{\,\,\,\, and \,\,\,\,}
( \sum_{i=1}^{\ell} |I_i \cap X_1|, \dots, \sum_{i=1}^{\ell} |I_i \cap X_{\tau}|) =  \vec{q} \text{.}
\end{align*}
The table is filled for increasing values of~$\ell$. As base case, we set~$T[\vec{p},\vec{q},1] := \varphi_1((f^{-1}(1) \setminus \vec{p}) \cup \vec{q})$ if~$(\vec{p},\vec{q}) \propto f^{-1}(1)$. Recall that~$(\vec{p'},\vec{q'}) \propto f^{-1}(1)$ is true if and only if~$\vec{p'}$ is subtractive compatible with~$f^{-1}(1)$ and~$\vec{q'}$ is additive compatible with~$f^{-1}(1)$.  If~$\vec{p}$ or~$\vec{q}$ is incompatible, we set~$T[\vec{p},\vec{q},1] := \infty$. The recurrence to compute an entry with~$\ell>1$~is
\begin{align*}
T[\vec{p}, \vec{q}, \ell] := \min_{\substack{\vec{p'} \leq \vec{p},~\vec{q'} \leq \vec{q} \\ (\vec{p'},\vec{q'}) \propto f^{-1}(\ell)}} T[\vec{p}-\vec{p'},\vec{q}-\vec{q'},\ell-1] \oplus \varphi_{\ell}((f^{-1}(\ell) \setminus \vec{p'}) \cup \vec{q'}) \text{,}
\end{align*}
Note that~$(\vec{0},\vec{0}) \propto f^{-1}(\ell)$ is always true, and therefore, such a minimum always exists. Herein,~$\vec{0}$ denotes the~$\tau$-dimensional vector where each entry is~$0$.

\begin{claim}

The recurrence is correct.
\end{claim}

\begin{claimproof}
Throughout the proof of this claim, we call the properties posed on the sets~$(R_1, \dots, R_\ell)$ and~$(I_1, \dots, I_\ell)$ in the definition of the table entries the \emph{desired properties for~$\vec{p}$, $\vec{q}$, and~$\ell$}.

We show that the recurrence is correct via induction over~$\ell$. If~$\ell=1$, we have~$T[\vec{p},\vec{q},1] := \varphi_1((f^{-1}(1) \setminus \vec{p}) \cup \vec{q})$. Thus, by \Cref{Definition: type vector operation}, there are sets~$R \subseteq f^{-1}(1)$ and~$I \subseteq X \setminus f^{-1}(1)$ where~$R$ is a set containing exactly~$p_j$ elements from~$X_j \cap f^{-1}(1)$ for every~$j  \in [1,\tau]$, and~$I$ is a set containing exactly~$q_j$ elements from~$X_j \cap (X\setminus f^{-1}(1))$ for every~$j  \in [1,\tau]$ and we have~$T[\vec{p},\vec{q},1]=\varphi_{1}((f^{-1}(1) \setminus R) \cup I)$. Since the value is invariant under the concrete choices of elements in~$R$ and~$I$ due to \Cref{Prop: TypeVectorOp Well-def}, the value $\varphi_{1}((f^{-1}(1) \setminus R) \cup I)$ is minimal among all such~$R$ and~$I$. Thus, the Recurrence holds for the base case~$\ell=1$.

Next, let~$\ell>1$ and assume that the recurrence holds for~$\ell-1$. Let~$(R_1, \dots, R_{\ell})$ and~$(I_1, \dots, I_{\ell})$ be sets with the desired properties for~$\vec{p}$, $\vec{q}$, and~$\ell$. We show
$$\bigoplus_{i=1}^{\ell} \varphi_i ((f^{-1}(i) \setminus R_i) \cup I_i) = T[\vec{p}, \vec{q}, \ell].$$

$(\geq)$ Since~$\oplus$ is associative and commutative, we have
\begin{align*}
\bigoplus_{i=1}^{\ell} \varphi_i ((f^{-1}(i) \setminus R_i) \cup I_i) = \underbrace{\left( \bigoplus_{i=1}^{\ell-1} \varphi_i ((f^{-1}(i) \setminus R_i) \cup I_i) \right)}_{=:Z} \oplus  \varphi_{\ell} ((f^{-1}(\ell) \setminus R_\ell) \cup I_\ell).
\end{align*}
Note that the vectors~$\vec{a}:= (|R_\ell \cap X_1|, \dots, |R_\ell \cap X_{\tau}|)$ and~$\vec{b}:= (|I_\ell \cap X_1|, \dots, |I_\ell \cap X_{\tau}|)$ satisfy~$\vec{a} \leq \vec{p}$ and~$\vec{b} \leq \vec{q}$. Moreover, since~$R_\ell \subseteq f^{-1}(\ell)$ and~$I_\ell \subseteq X \setminus f^{-1}(\ell)$, it holds that~$(\vec{a},\vec{b}) \propto f^{-1}(\ell)$ is true. Thus, the minimum of the right-hand-side of the recurrence includes~$\vec{a}$ and~$\vec{b}$, and by the definition of type specification operations and \Cref{Prop: TypeVectorOp Well-def}, we have~$\varphi(\ell) ((f^{-1} (\ell) \setminus \vec{a}) \cup \vec{q}) = \varphi(\ell) ((f^{-1} (\ell) \setminus R_\ell) \cup I_\ell)$. Since by the inductive hypothesis we have~$T[\vec{p}-\vec{a}, \vec{q}- \vec{b}, \ell-1] \leq Z$, we conclude~$\bigoplus_{i=1}^{\ell} \varphi_i ((f^{-1}(i) \setminus R_i) \cup I_i) \geq T[\vec{p}, \vec{q}, \ell]$.

$(\leq)$ Let~$\vec{a}$ and~$\vec{b}$ be the vectors minimizing the right-hand-side of the recurrence. By the definition of type specification operations, there are sets~$\widetilde{R_\ell} \subseteq f^{-1}(\ell)$ and~$\widetilde{I_\ell} \subseteq X \setminus f^{-1}(\ell)$ with~$(|\widetilde{R_\ell} \cap X_1|, \dots, |\widetilde{R_\ell} \cap X_{\tau}|) = \vec{a}$ and~$(|\widetilde{I_\ell} \cap X_1|, \dots, |\widetilde{I_\ell} \cap X_{\tau}|) = \vec{b}$. By inductive hypothesis, the value~
$T[\vec{p}-\vec{a}, \vec{q}- \vec{b}, \ell-1]$ corresponds to the minimal value of~$\bigoplus_{i=1}^{\ell-1} \varphi_i ((f^{-1}(i) \setminus \widetilde{R_i}) \cup \widetilde{I_i})$ for sets~$(\widetilde{R_1}, \dots, \widetilde{R_{\ell-1}})$ and~$(\widetilde{I_1}, \dots, \widetilde{I_{\ell-1}})$ that satisfy the desired properties for~$\vec{p}-\vec{a}$, $\vec{q}-\vec{b}$, and~$\ell-1$. Then,~$T[\vec{p},\vec{q},\ell]$ corresponds to the value~$\bigoplus_{i=1}^{\ell} \varphi_i ((f^{-1}(i) \setminus \widetilde{R_i}) \cup \widetilde{I_i})$ for sets~$(\widetilde{R_1}, \dots, \widetilde{R_{\ell}})$ and~$(\widetilde{I_1}, \dots, \widetilde{I_{\ell}})$ that satisfy the desired properties for~$\vec{p}$, $\vec{q}$, and~$\ell$. Thus, it follows that~$\bigoplus_{i=1}^{\ell} \varphi_i ((f^{-1}(i) \setminus R_i \cup I_i)) \leq T[\vec{p}, \vec{q}, \ell]$.
\end{claimproof}

\textit{Computation of a solution~$f'$.} We next describe how to compute the desired~$b$-partition~$f'$ after filling the dynamic programming table~$T$. The sets~$(R_1, \dots, R_b)$ and~$(I_1, \dots, I_b)$ corresponding to the table entry~$T[\vec{\delta}, \vec{\delta}, b]$ can be found via trace back. We let~$\mathcal{R}:= \bigcup_{i=1}^{b} R_i$. Due to the property~$R_i \subseteq f^{-1}(i)$ for all~$i \in [1,b]$, all~$R_i$ are disjoint. Together with the properties on the vectors~$\vec{p}$ and~$\vec{q}$, this implies~$
|\mathcal{R} \cap X_j| = \sum_{i=1}^{b} |R_i \cap X_j| = \sum_{i=1}^{b} |I_i \cap X_j| = \delta_j$
for every~$j \in [1,\tau]$. 
Consequently, for every~$j \in [1,\tau]$, there is a mapping~$\gamma_j: \mathcal{R} \cap X_j \rightarrow [1, b]$, that maps exactly~$|I_i \cap X_j|$ elements of~$\mathcal{R} \cap X_j$ to~$i$ for every~$i \in [1,b]$. Intuitively, the mapping~$\gamma_j$ assigns a new bin to each item of type~$j$ that was removed from some bin.

The~$b$-partition~$f'$ is defined via the mappings~$\gamma_j$ as follows: For every~$x \in X \setminus \mathcal{R}$, we set~$f'(x):=f(x)$. For every~$x \in \mathcal{R}$, we set~$f'(x) := \gamma_j (x)$ for the corresponding~$j$ with~$x \in \mathcal{R} \cap X_j$. 

By the definition of~$f'$, the~$b$-partitions $f$ and~$f'$ may only differ on~$\mathcal{R}$ and therefore,~$|\mathcal{R} \cap X_j| = \delta_j$ implies~$
(|X_1 \cap D_{\rm flip}(f,f')|, \dots, |X_\tau \cap D_{\rm flip}(f,f')|) \leq \vec{\delta}$.
It remains to show that~$\val(f')$ is minimal. By the construction of~$f'$ we have~$f'^{-1}(i) = f^{-1}(i) \setminus R_i \cup I'_i$, where~$I'_i := \bigcup_{j=1}^{\tau} \{ x \in \mathcal{R} \cap X_j \mid \gamma_j(x) = i \}$.
Then, by the construction of the~$\gamma_j$ we have
\begin{align*}
(|I'_i \cap X_1| , \dots, |I'_i \cap X_\tau|)
=&\, (|\{ x \in \mathcal{R} \cap X_1 \mid \gamma_1(x) = i \}| , \dots, |\{ x \in \mathcal{R} \cap X_{\tau} \mid \gamma_{\tau}(x) = i \}|)\\
=&\, (|I_i \cap X_1| , \dots, |I_i \cap X_{\tau}|).
\end{align*}
Thus, the sets~$(f^{-1}(i) \setminus R_i) \cup I'_i$ and~$(f^{-1}(i) \setminus R_i) \cup {I}_i$ contain the exact same number of elements from each~$X_j$. Consequently, we have~$\varphi_i( (f^{-1}(i) \setminus R_i) \cup I'_i) = \varphi((f^{-1}(i) \setminus R_i) \cup {I}_i)$ for every~$i \in [1,b]$ and thus, by the definition of types, the minimality of~$T[\vec{\delta},\vec{\delta},b]$ implies the minimality of~$\val(f')$.

\textit{Running time.} We finally provide two different ways to analyze the running time.

$a)$ The vector~$\vec{\delta}$ can be regarded as a set~$M$ containing~$\sum_{j=1}^{\tau} \delta_j = k$ individual identifiers from which exactly~$\delta_j$ are labeled with type~$j$ for each~$j \in [1,\tau]$. With this view on~$\vec{\delta}$, the vectors~$\vec{p} \leq \vec{\delta}$ and~$\vec{q} \leq \vec{\delta}$ correspond to subsets of~$M$, so we have~$2^k$ possible choices for~$\vec{p}$ and~$2^k$ possible choices for~$\vec{q}$. 
Consequently, the size of the table~$T$ is~$2^{2k} \cdot b$. To compute one entry, one needs to consider up to~$2^{2k}$ choices of~$\vec{p'}$ and~$\vec{q'}$. For each such choice,~$(\vec{p'},\vec{q'}) \propto f^{-1}(\ell)$ can be checked in~$|X|^{\Oh(1)}$~time. With the evaluation of the~IBE, this leads to a total running time of~$2^{4k} \cdot b \cdot \Phi \cdot |X|^{\Oh(1)}$.

$b)$ The number of possible vectors that are component-wise smaller than~$\vec{\delta}$ is~$\prod_{j=1}^{\tau} (\delta_j + 1)$, since each entry is an element of~$[0,\delta_j]$.
Since~$\sum_{j=1}^{\tau} \delta_j = k$, the product of all~$(\delta_j +1)$ is maximal if all~$\delta_j$ have roughly the same size~$\frac{k}{\tau}$. Thus, we have~$\prod_{j=1}^{\tau} (\delta_j + 1) \leq (\frac{k}{\tau} + 1)^{\tau}$.
Consequently, the size of~$T$ is upper-bounded by~$(\lceil \frac{k}{\tau} \rceil + 1)^{2\tau} \cdot b$. To compute one entry, one needs to consider up to~$(\lceil \frac{k}{\tau} \rceil + 1)^{2\tau}$ choices of~$\vec{p'}$ and~$\vec{q'}$. For each such choice,~$(\vec{p'},\vec{q'}) \propto f^{-1}(\ell)$ can be checked in~$|X|^{\Oh(1)}$~time. With the evaluation of the~IBE, this leads to a total running time of~$(\lceil \frac{k}{\tau} \rceil + 1)^{4\tau} \cdot b \cdot \Phi \cdot |X|^{\Oh(1)}$.
\end{proof}

We next use \Cref{Lemma: DP Algorithm} to prove \Cref{Theorem: Framework Algorithms}. 

\begin{proof}[Proof of \Cref{Theorem: Framework Algorithms}]
Recall that the idea is to consider every possible vector~$\vec{\delta}$ specifying how many elements of which type belong to the elements of the flip. For each possible~$\vec{\delta}$ we then use the algorithm behind \Cref{Lemma: DP Algorithm}. It remains to describe how to efficiently iterate over the possible choices of~$\vec{\delta}$ to obtain the running times~$a)$ and~$b)$.

$a)$ Let~$\vec{e_1}, \dots, \vec{e_\tau}$ be the~$\tau$-dimensional unit vectors.
That is, only the~$j$th entry of~$\vec{e_j}$ equals~$1$ and all other entries equal~$0$. We enumerate all~$\vec{\delta}$ with~$\sum_{i=1}^{\tau} \delta_i \leq k$ by considering all~$\tau^{k}$ possibilities to repetitively draw up to~$k$-times from the set~$\{\vec{e_1}, \dots, \vec{e_\tau}\}$.
For each such choice~$\vec{\delta}$, we check whether~$\delta_i \leq |X_i|$ for each~$i \in [1,\tau]$ to ensure that~$\vec{\delta}$ is a type specification. If this is the case, we apply the algorithm behind \Cref{Lemma: DP Algorithm}. With the running time from \Cref{Lemma: DP Algorithm}~$a)$, this leads to a total running time of~$\tau^k \cdot 2^{\Oh(k)} \cdot b \cdot \Phi \cdot |X|^{\Oh(1)}$.

$b)$ Note that for a vector~$\vec{\delta}$ with~$\sum_{i=1}^{\tau} \delta_i = k$, we have~$\delta_i \in [0,k]$ for every~$i \in [1,\tau]$. Thus, in time~$(k+1)^{\tau} \cdot k^{\Oh(1)}$ we can enumerate all possible~$\vec{\delta}$. Analogously to~$a)$, we check whether~$\delta_i \leq |X_i|$ for each~$i \in [1,\tau]$. With the running time from \Cref{Lemma: DP Algorithm}~$b)$, this leads to a total running time of~$k^{\Oh(\tau)} \cdot b \cdot \Phi \cdot |X|^{\Oh(1)}$.
\end{proof}

\section{Applications of the Framework} \label{Section: Framework Application}
We next study parameterized local search versions of a wide range of well-known problems. For each of these problems, we consider parameterization by the search radius in combination with some structural parameter. 
All results rely on the algorithm behind \Cref{Theorem: Framework Algorithms}. Therefore, it suffices to express an instance~$I$ of a local search problem as a corresponding instance~$J = (b, X, (\varphi_i)_{i \in [1,b]}, f,k)$ of \LSGBP. The proofs in this section are given in the following structure:
\begin{enumerate}[1.]
\item \emph{GBP Construction:} Given the instance~$I$ of the local search problem, we specify the universe~$X$, the number of bins~$b$ and the IBE~$(\varphi_i)_{i \in [1,b]}$ of the instance~$J$. Since the running time~$\Phi$ to evaluate the IBE is part of the running time in \Cref{Theorem: Framework Algorithms}, the GBP construction also includes a description of how to efficiently evaluate the IBE from the input instance~$I$.
\item \emph{Type Partition:} We express how a type partition for~$J$ is computed from~$I$ and describe how the targeted structural parameter corresponds to the number of types as defined in \Cref{Section: Framework}.
\item \emph{Solution Correspondence:} Recall that an instance of a local search problem always contains a solution of the underlying problem. Let~$s$ be the solution given in the instance~$I$. To show the solution correspondence, we describe how solutions of~$I$ can be transformed into solutions of~$J$ and vice versa, such that the following holds: A solution~$s'$ of~$I$ is \emph{better} than the given solution~$s$ of~$I$ if and only if the corresponding $b$-partition~$f_{s'}$ has strictly smaller (larger) target value than~$f_{s}$ with respect to the~IBE~$(\varphi_i)_{i \in [1,b]}$.
\end{enumerate}

\subsection{Max $c$-Cut}
\label{sec:MC}
Let~$c \in \nn$ and let~$G=(V,E)$ be an undirected graph.
We say that an edge~$e\in E$ is \emph{properly colored} by a~$c$-partition~$\colo$ of~$V$, if~$\colo$ assigns distinct colors to the endpoints of~$e$.

\defproblem{\MC}{An undirected graph~$G=(V,E)$.}{Find a~$c$-partition~$\colo$ of~$V$ that maximizes the number of properly colored edges under~$\colo$.}

Note that the goal of~\MC can also be equivalently redefined as: 
Find a~$c$-partition~$\colo$ of~$V$ that minimizes the number of edges that are not properly colored under~$\colo$, that is, a~$c$-partition~$\colo$ that minimizes~$\faults(\chi) := \sum_{i\in [1,c]} |E_G(\colo^{-1}(i))|$.

The input of the corresponding local search problem \LMC additionally consists of a~$c$-partition~$\chi$ and some~$k \in \mathds{N}$, and one aims to find a~$c$-partition~$\chi'$ with~$\faults(\chi') < \faults(\chi)$ and~$\dflip(\chi,\chi') \leq k$.

\LMC is W[1]-hard parameterized by~$k$~\cite{GGKM23} and cannot be solved in $f(k)\cdot n^{o(k)}$~time unless the ETH is false~\cite{M24}.
Form a positive side, it was shown that the problem can be solved in $2^{\Oh(k)}\cdot |I|^{\Oh(1)}$~time on apex-minor-free graphs~\cite{FFLRSV12}.
Moreover, on general graphs, the problem can be solved in~$\Delta^{\Oh(k)}\cdot |I|^{\Oh(1)}$~time~\cite{GGKM23}, where~$\Delta$ denotes the maximum degree of the input graph.
This algorithms found successful application as a post-processing algorithm for a state-of-the-art heuristic for~\MC~\cite{GGKM23}.

We show the following.

\begin{theorem} \label{Theorem: Algo Max-c-Cut}
\LMC can be solved in~$\nd^k \cdot 2^{\Oh(k)} \cdot n^{\Oh(1)}$ and~$k^{\Oh(\nd)} \cdot n^{\Oh(1)}$~time.
\end{theorem}

\begin{proof}
The proof relies on the algorithm behind \Cref{Theorem: Framework Algorithms}. We describe how to use this algorithm to solve an instance~$I:=(G=(V,E), \chi, k)$ of \LMC.

\textit{1. GBP Construction:} We describe how to obtain an instance~$J:=(b,X,(\varphi_i)_{i \in [1,b]},f,k)$ of \LSGBP. Our universe~$X$ is exactly the vertex set~$V$ of~$G$, and the number~$b$ of bins is exactly the number~$c$ of colors. We next define the IBE. For each~$i \in [1,c]$ and each vertex set~$S \subseteq V$, we define~$\varphi_i(S) := |E_G(S)|$. Note that for each value~$\varphi(S)$ can obviously be computed in~$n^{\Oh(1)}$ time, and thus, the IBE can be evaluated in polynomial time from the input instance~$I$. Our operation~$\oplus$ is the sum of integer numbers.

\textit{2. Type Partition:} Our type partition is the collection of neighborhood classes of~$G$ that can be computed in~$\Oh(n+m)$~time~~\cite{HM91}. 

To show that this collection is in fact a type partition, we show that two vertices from the same neighborhood class are target equivalent according to \Cref{Definition: Types}. To this end, let~$C$ be a neighborhood class of~$G$, let~$u$ and~$v$ be vertices of~$C$.
We show that~$u$ and~$v$ are target equivalent. 
Let~$S\subseteq V$ be a vertex set with~$S \cap \{u,v\} = \{u\}$.
We show that for each~$i\in [1,b]$, $\varphi_i(S) = \varphi_i((S \setminus \{u\}) \cup \{v\})$.
Let~$S' := (S \setminus \{u\}) \cup \{v\}$.
Since all IBEs~$\IBEs$ are identically defined, it suffices to only consider~$i=1$.
By the fact that~$C$ is a neighborhood class of~$G$, $u$ and~$v$ have the same neighbors in~$S\setminus \{u\} = S' \setminus \{v\}$.
Hence, $\varphi_1(S) = \varphi_1(S')$.
This implies that~$u$ and~$v$ are target equivalent.

\textit{3. Solution Correspondence:} Note that the set of solutions~$\chi'$ of \LMC is exactly the set of all~$c$-partitions of~$V$. Since~$X=V$ and~$b=c$, the solutions of~$I$ have a one-to-one correspondence to the solutions of~$J$. Moreover, note that the definition of the IBE is based on the reformulation of the objective function of~\MC. Since~$\oplus$ is the sum of integer numbers, we have~$\val(\chi') = \faults(\chi')$ for every~$c$-partition~$\chi'$. Therefore,~$\chi'$ is a better solution than~$\chi$ for \LMC if and only if~$\chi'$ is a smaller target value than~$\chi$ with respect to the IBE.
\end{proof}

\subsection{Nash Social Welfare}\label{sec:NSW} 
Let~$n,m$ be integers.
In this section, $\mathcal{A}$ is a set of $n$~\emph{agents} and~$\mathcal{S}$ is a set of $m$~\emph{items}.
Furthermore, we have~$n$ \emph{utility functions}~$(u_i:\mathcal{S}\to\mathds{N})_{i\in[1,n]}$. 
An $n$-partition~$\colo$ of~$\mathcal{S}$ is called an \emph{allocation}.
For an allocation~$\chi$, the function~$
\Nash(\chi) := \prod_{i\in[1,n]} \left( \sum_{s \in \chi^{-1}(i)} u_i(s) \right)$
is called the \emph{Nash score} of~$\chi$.

\defproblem{\NSW}
{A set~$\mathcal{A}$ of agents, a set~$\mathcal{S}$ of items, and a set~$(u_i)_{i\in[1,n]}$ of utilities.}
{Find an allocation~$\chi$ that maximizes~$\Nash(\chi)$.}

The input of the corresponding local search problem \NSWLoc additionally consists of an allocation~$\chi$ and some~$k \in \mathds{N}$ and one asks for an allocation~$\chi'$ with~$\Nash(\chi') > \Nash(\chi)$ and~$\dflip(\chi,\chi') \leq k$.

We say that two items~$j_1$ and~$j_2$ are \emph{identical} if each agents values~$j_1$ and~$j_2$ similarly, that is, if~$u_i(j_1)=u_i(j_2)$ for each~$i\in[n]$. Two items are \emph{distinct} if they are not identical. 

To the best of our knowledge, parameterized local search for \NSW has not yet been studied. We now show the following.

\begin{theorem} \label{Theorem: Algo Nash}
Let~$\tau$ denote the maximum number of pairwise distinct items in an instance of \NSWLoc. \NSWLoc can be solved in $\tau^k \cdot 2^{\Oh(k)} \cdot (n+m)^{\Oh(1)}$~time and in $k^{\Oh(\tau)} \cdot (n+m)^{\Oh(1)}$~time.
\end{theorem}

\begin{proof}
The proof relies on the algorithm behind \Cref{Theorem: Framework Algorithms}. We describe how to use this algorithm to solve an instance~$I:=(\mathcal{A}, \mathcal{S},(u_i)_{i\in[1,n]}, \chi,k)$ of \NSWLoc.

\textit{1. GBP construction:}
We describe how~$I$ corresponds to an instance~$J:= (b,X, (\varphi_i)_{i \in [1,b]},f,k)$ of \LSGBP. Our universe~$X$ is exactly the set~$\mathcal{S}$ of items and~$b$ is exactly the number~$n$ of agents. It remains to define the IBE. For~$i \in [1,n]$ we let~$\varphi_i(X') := \sum_{s \in X'} u_i(s)$. Given the instance~$I$, each value~$\varphi_i(X')$ can obviously be computed in polynomial time in~$|I|$, where~$|I|$ denotes the total encoding size of the instance~$I$. Our operation~$\oplus$ is the multiplication of integer numbers.

\textit{2. Type Partition:}
Let~$(X_1, \dots, X_{\tau})$ be a partition of~$\mathcal{S}$ where each~$X_j$ is an inclusion-wise maximal set of pairwise identical elements. Note that this partition can clearly be computed in~$(n+m)^{\Oh(1)}$~time. Moreover, since identical elements have the same values for each agent, the elements of each~$X_j$ are pairwise target-equivalent according to \Cref{Definition: Types}.

\textit{3. Solution Correspondence:}
Since~$\oplus$ is the multiplication of integer numbers and~$\varphi_i(X') := \sum_{s \in X'} u_i(s)$ for each~$i \in  [1,n]$, we have~$\val(\chi) = \Nash (\chi)$ for every allocation~$\chi$.
\end{proof}

In \Cref{NSWisHard} we complement the above algorithm by showing that \NSWLoc is W[1]-hard parameterized by~$k$ alone and that an algorithm with running time~$n^{o(k)}$ would refute the ETH.

\subsection{Cluster Editing}
\label{sec:CE}
A \emph{cluster graph} is an undirected graph in which each connected component forms a clique. Let~$G=(V,E)$ be an undirected graph. In \CE, one aims to apply a minimum number of edge modification (edge insertions and edge deletions) on~$G$, such that the resulting graph is a cluster graph. We let~$\binom{V}{2}$ denote the set of two-element subsets of vertices corresponding to possible edge modifications. Given a set~$E' \subseteq \binom{V}{2}$, we let~$E \triangle E' := (E \setminus E') \cup (E' \setminus E)$ denote the \emph{symmetric difference} corresponding to the application of the graph modifications. 

\defproblem{\CE}{An undirected graph~$G=(V,E)$.}{Find a set~$E' \subseteq \binom{V}{2}$ such that~$(V,E \triangle E')$ is a cluster graph and~$|E'|$ is minimal under this property.}

Note that the maximum number of clusters corresponds to the number of vertices~$n$. We consider an equivalent definition of \CE where one asks for an~$n$-labeling~$\chi$ with partitioning~$V$ into~$n$ (possibly empty) classes~$\chi^{-1}(1), \dots, \chi^{-1}(n)$, such that
\begin{align*}
\score(\chi) := \sum_{i=1}^{n} |E(\chi^{-1}(i))| - ( \binom{|\chi^{-1}(i)|}{2} - |E(\chi^{-1}(i))|)
\end{align*}
is maximal~\cite{M24}. Intuitively, maximizing this score corresponds to maximizing the number of present edges in the connected components minus the required edge insertions such that these components form a cliques.

We study a parameterized local search version of \CE where one flip in the flip neighborhood corresponds to moving single vertices from their current cluster into another cluster. 
The input of the corresponding local search problem \LSCE consists of the input of \CE, together with an additional integer~$k$ and an~$n$-labeling~$\chi$ of~$V$ for some. The task is to compute a~$n$-labeling~$\chi'$ such that~$\score(\chi')> \score(\chi)$ and~$\dflip(\chi,\chi') \leq k$. 

With respect to this neighborhood, \LSCE is known to be W[1]-hard parameterized by~$k$ and cannot be solved in $f(k)\cdot n^{o(k)}$~time unless the ETH is false~\cite{GMNW23}.
Form a positive side, it was shown that the problem can be solved in $\Delta^{\Oh(k)}\cdot |I|^{\Oh(1)}$~time~\cite{GMNW23}, where~$\Delta$ denotes the maximum degree of the input graph.

We now show the following.

\begin{theorem} \label{Theorem: Algo Cluster Editing}
\LSCE can be solved in~$\nd^k \cdot 2^{\Oh(k)} \cdot n^{\Oh(1)}$~time and in~$k^{\Oh(\nd)} \cdot n^{\Oh(1)}$~time.
\end{theorem}

\begin{proof}
The proof relies on the algorithm behind \Cref{Theorem: Framework Algorithms}. We describe how to use this algorithm to solve an instance~$I:=(G=(V,E),k,\chi)$. 

\textit{1. GBP Construction:} We describe how~$I$ corresponds to an \LSGBP-instance $J:=(b,X,(\varphi_i)_{i \in [1,b]},f,k)$. Our universe~$X$ is exactly the vertex set~$V$ of~$G$, and the number~$b$ of bins is exactly~$n$. We next define the IBE. For each~$i \in [1,n]$ and each vertex set~$S \subseteq V$, we define~$\varphi_i(S):=|E(S)| - (\binom{|S|}{2} - |E(S)|)$. Note that each value~$\varphi_i(S)$ can obviously be computed in~$n^{\Oh(1)}$~time from the input graph~$G$. Our operation~$\oplus$ is the sum of integer numbers.

\textit{2. Type Partition:} Our type partition is the collection of neighborhood classes of~$G$. Recall that this collection can be computed in~$\Oh(n+m)$~time.

To show that this collection is in fact a type partition, we show that two vertices from the same neighborhood class of~$G$ are target equivalent according to \Cref{Definition: Types}. To this end, let~$C$ be a neighborhood class of~$G$ and let~$u$ and~$v$ be vertices of~$c$. Let~$S \subseteq V$ be a vertex set with~$S \cap \{u,v\} = \{u\}$. We obviously have~$|S| = |(S \setminus \{u\}) \cup \{v\}|$. Moreover, since~$u$ and~$v$ have the exact same neighbors in~$G$, we have~$|E(S)| = |E((S \setminus \{u\}) \cup \{v\})|$. Then, by the definition of the IBE we have~$\varphi_i(S) = \varphi_i((S \setminus \{u\}) \cup \{v\})$ for all~$i \in [1,n]$.

\textit{3. Solution Correspondence:} Note that the set of solutions~$\chi'$ of \LSCE is exactly the set of all~$n$-partitions of~$V$. Since~$X=V$ and~$b=n$, the solutions of~$J$ have a one-to-one correspondence to the solutions of~$J$. Moreover, note that the definition of the IBE is based on the reformulation of the objective function of~\CE and~$\oplus$ is the sum of integer numbers, we have~$\val(\chi') = \score(\chi')$ for every~$n$-partition~$\chi'$. Therefore,~$\chi'$ is a better solution than~$\chi$ for~\CE if and only if~$\chi'$ has a larger target value than~$\chi$ with respect to the IBE.
\end{proof}

\subsection{Vector Bin Packing}
\label{sec:VBP}

Let~$b,d$ be integers. 
In this section for~$i\in[1,b]$, $B_i$ is a \emph{bin}.
Each bin~$B_i$ is associated with a \emph{weight vector}~$\omega_i\in\mathds{N}^d$.
Let~$\mathcal{S}$ be a set of vectors from~$\mathds{N}_0^d$. A $b$-partition~$\chi$ of~$\mathcal{S}$ is called a \emph{bin assignment}.

Given a subset~$S \subseteq \mathcal{S}$, we say that the \emph{overload of bin~$B_i$ with respect to~$S$} is~$\oh(B_i,S) := \sum_{j=1}^d \max{(0,(\sum_{v \in S} v_j) - w_j)}$. Our target is to find an assignment~$\chi$ minimizing the total overload, which is defined as $\sum_{i\in [1,b]} \oh(B_i,\chi^{-1}(i))$. 

\defproblem{\Bin}
{Integers~$b,d$, a set~$\mathcal{S}$ of vectors from~$\mathds{N}_0^d$, a vector~$\omega_i\in\mathds{N}_0^d$ for each~$i\in[1,b]$.}
{Find a bin assignment~$\chi$ of~$\mathcal{S}$ such that the total overload is minimized, that is, $\sum_{i\in [1,b]} \oh(B_i,\chi^{-1}(i))$ is minimized.}

The input of the corresponding local search problem \BinLoc additionally consists of a bin assignment~$\chi$ and some~$k \in \mathds{N}$, and one aims to find a bin assignment~$\chi'$ with~$\sum_{i\in b} \oh(B_i,\chi'^{-1}(i)) < \sum_{i\in b} \oh(B_i,\chi^{-1}(i)))$ and~$\dflip(\chi,\chi') \leq k$.

We say that two vectors~$\vect_1$ and~$\vect_2$ are \emph{identical} if they agree in all dimensions, that is, if~$\vect_1(j)=\vect_2(j)$ for each dimensions~$j\in[1,d]$ and two vectors are \emph{distinct} otherwise.

To the best of our knowledge, parameterized local search for \Bin has not yet been studied. We now show the following.

\begin{theorem} \label{Theorem: Algo Vector Bin Packing}
Let~$\tau$ denote the maximum number of pairwise distinct items in an instance of \BinLoc. \BinLoc can be solved in~$\tau^k \cdot 2^{\Oh(k)} \cdot (b+d+|\mathcal{S}|)^{\Oh(1)}$~time and in~$k^{\Oh(\tau)} \cdot (b+d+|\mathcal{S}|)^{\Oh(1)}$~time.
\end{theorem}

\begin{proof}
The proof relies on the algorithm behind \Cref{Theorem: Framework Algorithms}. We describe how to use this algorithm to solve an instance~$I := (b,d,\mathcal{S}, (\omega_i)_{i \in [1,b]}, \chi,k)$ of \BinLoc.

\textit{1. GBP construction:} Initially, we describe how~$I$ corresponds to an instance $J:=(b,X,(\varphi_i)_{i \in [1,b]},f,k)$ of \LSGBP. Our universe~$X$ is exactly the set~$\mathcal{S}$ and the number~$b$ is the number of bins in~$I$. It remains to define the~IBE. For~$i \in [1,b]$, we define~$\varphi_i(X'):= \oh(B_i,X')$. That is, the~$i$th IBE corresponds to the overload of~$B_i$ with respect to the vectors~$X'$. Our operation~$\oplus$ is the sum of integer numbers.

\textit{2. Type Partition:} Let~$(X_1, \dots, X_{\tau})$ be a partition of~$\mathcal{S}$ where each~$X_j$ is an inclusion-wise maximal set of pairwise identical elements. Note that this partition can clearly be computed in~$|\mathcal{S}|^{\Oh(1)}$~time.

To show that this collection is in fact a type partition, we show that two elements from one set~$X_j$ are target equivalent according to \Cref{Definition: Types}.
Let~$\vect_1$ and~$\vect_2$ be two vectors of the same set~$X_j$.
Let~$S\subseteq \mathcal{S}$ be a set with~$S \cap \{\vect_1, \vect_2\} = \{\vect_1\}$.
Let~$S' \coloneqq (S \setminus \{\vect_1\}) \cup \{\vect_2\}$.
We show that for each~$i\in [1,b]$, $\varphi_i(S) = \varphi_i(S')$.
Without loss of generality, we consider the IBE~$\varphi_i$.
Since~$\vect_1$ and~$\vect_2$ agree in all dimensions, we have~$S\setminus \{\vect_1\} = S' \setminus \{\vect_2\}$.
Consequently, $\varphi_i(S\cup \{\vect_2\}) = \varphi_i(S' \cup \{\vect_1\})$ and thus~$\vect_1$ and~$\vect_2$ are target equivalent. 

\textit{3. Solution Correspondence:} Since~$\oplus$ is the sum of integer numbers and~$\varphi_i(X') = \oh(B_i,X')$ for all~$i \in [1,b]$, we have~$\val(\chi) = \sum_{i\in b} \oh(B_i,\chi^{-1}(i))$ for every bin assignment~$\chi$.
\end{proof}

In \Cref{prop-ls-bin-packing} we complement the above algorithm by showing that \BinLoc is W[1]-hard parameterized by~$k$ and cannot be solved in $f(k)\cdot n^{o(k)}$~time unless the ETH fails.

\subsection{Multi Knapsack}
\label{sec:MK}
Recall that by Definition~\ref{Definition: Target Function} it is allowed that a function~$\varphi_i$ evaluates to~$\infty$ or~$-\infty$. While the previous examples only use IBEs evaluating to finite values, we now use~$-\infty$ to mark infeasible solutions for instances of \textsc{Multi Knapsack}. 

Let~$m \in \mathds{N}$ be a number of knapsacks, with their capacities~$W_1, \dots, W_m \in \mathds{N_0}$ and let~$[1,n]$ be a set of items with values~$v_{(\ell,i)} \in \mathds{N_0}$ and weights~$w_{(\ell,i)} \in \mathds{N_0}$ for all~$\ell \in [1,n]$ and~$i \in [1,m]$. We say that an~$(m+1)$-partition~$\chi$ of~$[1,n]$ \emph{fits into the knapsacks}, if~$\sum_{i\in\chi^{-1}(\ell)} w_{(\ell,i)} \leq W_i$ for every~$i \in [1,m]$. We define the \emph{score of~$\chi$} as~$\score(\chi) := \sum_{i=1}^{m} \sum_{\ell \in \chi^{-1}(i)} v_{(\ell,i)}$. Intuitively, for~$i \in [1,m]$, the set~$\chi^{-1}(i)$ corresponds to the chosen items for the~$i$th knapsack and~$\chi^{-1}(m+1)$ corresponds to the set of non-chosen items.

\defproblem{\MK}{Knapsacks with capacities~$W_1, \dots, W_m \in \mathds{N_0}$, items~$[1,n]$ with values~$v_{(\ell,i)} \in \mathds{N_0}$ and weights~$w_{(\ell,i)} \in \mathds{N_0}$ for all~$\ell \in [1,n]$ and~$i \in [1,m]$.}{Find an~$(m+1)$-partition~$\chi$ of~$[1,n]$ that fits into the knapsacks, such that~$\score(\chi)$ is maximal.}

The input of the corresponding local search problem~\MKLoc additionally consists of an~$m$-partition~$\chi$ that fits into the knapsacks and some~$k \in \mathds{N}$ and one asks for an~$(m+1)$-partition~$\chi'$ with~$\dflip(\chi, \chi') \leq k$ that fits into the knapsacks and satisfies~$\score(\chi')>\score(\chi)$.

We say that two items~$\ell \in [1,n]$ and~$\ell' \in [1,n]$ are \emph{identical} if~$v_{(\ell,i)} = v_{(\ell',i)}$ and~$w_{(\ell,i)} = w_{(\ell',i)}$ for all~$i \in [1,m]$. Two elements are distinct if they are not identical. 

To the best of our knowledge, parameterized local search for \MK has not yet been studied. We now show the following.

\begin{theorem}
 \label{Theorem: Algo Multi-Knapsack}
Let~$\tau$ denote the maximum number of pairwise distinct items in an instance of \MKLoc. \MKLoc can be solved in~$\tau^k \cdot 2^{\Oh(k)} \cdot (n+m)^{\Oh(1)}$~time and in~$k^{\Oh(\tau)} \cdot (n+m)^{\Oh(1)}$~time.
\end{theorem}

\begin{proof}

The proof relies on the algorithm behind \Cref{Theorem: Framework Algorithms}. We describe how to use this algorithm to solve an instance
$$I := ( (W_i)_{i \in [1,m]}, (v_{(\ell,i)})_{\ell \in [1,n],i \in [1,m]},  (w_{(\ell,i)})_{\ell \in [1,n],i \in [1,m]}, \chi, k)$$
of \MKLoc.

\textit{1. GBP Construction:} 
We describe how~$I$ corresponds to an instance $J:=(b,X,(\varphi_i)_{i \in [1,b]},f,k)$ of~\LSGBP.
Our universe~$X$ is exactly the set~$[1,n]$ and the number of bins~$b$ is exactly~$m+1$, which is the number of knapsacks~$m$ plus one set for the non-chosen items. It remains to define the~IBE~$(\varphi_i)_{i \in [1,m+1]}$. Intuitively, for~$i \in [1,m]$, the mapping~$\varphi_i$ assigns the total value received from a set of items assigned to the knapsack with the capacity~$W_i$. However, if the total weight exceeds~$W_i$, the value of~$\varphi_i$ equals~$-\infty$.  Formally, for every~$i \in [1,m]$ we define~$\varphi_i: 2^X \rightarrow \mathds{Z}$~by
\begin{align*}
\varphi_i(X') :=
\begin{cases}
 -\infty & \text{ if } \sum_{\ell \in X'} w_{(\ell,i)} > W_j, \text{ or}\\
 \sum_{\ell \in X'} v_{(\ell,i)} & \text{otherwise.}
 \end{cases}
\end{align*}
Furthermore, we let~$\varphi_{m+1}(X'):=0$ for every~$X' \subseteq X$. Given the instance~$I$, one value~$\varphi_i(X')$ can obviously be computed in polynomial time in~$|I|$, where~$|I|$ denotes the total encoding size of the instance~$I$. Our operation~$\oplus$ is the sum of integer numbers.

\textit{2. Type Partition:} Next, let~$(X_1, \dots, X_{\tau})$ be a partition of~$[1,n]$ where each~$X_j$ is an inclusion-wise maximal set of pairwise similar elements. This partition can be computed form~$I$ in~$\Oh(n+m)$-time. Furthermore, by the definition of the IBE, the elements in each~$X_j$ are pairwise target-equivalent according to \Cref{Definition: Types}. Consequently,~$(X_1, \dots, X_{\tau})$ is a type partition. This completes our description of the instance~$J$ of \LSGBP together with its type partition.

\textit{3. Solution Correspondence:}
Note that the set of solutions~$\chi'$ of \MKLoc is exactly the set of all~$(m+1)$-partitions of~$[1,n]$. Since~$X=[1,n]$ and~$b=m+1$, the solutions of~$J$ have a one-to-one correspondence to the solutions of~$I$. We next establish in which sense the IBE corresponds to the score of MK.

\begin{claim} \label{Claim: Multi-Knapsack IBE}
Let~$\chi$ be an~$(m+1)$-partition of~$[1,n]$. Then,~$\chi$ fits into the knapsacks if and only if~$\val(\chi) \neq -\infty$. Moreover, if~$\chi$ fits into the knapsacks, we have~$ \val(\chi) = \score(\chi) \geq 0$.
\end{claim}

\begin{claimproof}
If~$\chi$ does not fit into the knapsacks, we have~$\varphi_i(\chi^{-1}(i)) = -\infty$ for some~$i \in [1,m]$ and therefore~$\val(\chi) = -\infty$. Otherwise, if~$\chi$ fits into the knapsacks, we have~$\varphi_i(\chi^{-1}(i)) = \sum_{\ell \in \chi^{-1}(i)} v_{(\ell,i)}$ for all~$i \in [1,m]$ and~$\varphi_{m+1}(\chi^{-1}(m+1))=0$. Thus,~$ \val(\chi) = \score(\chi) \geq 0$.
\end{claimproof}

Next let~$\chi$ be the input solution of the instance~$I$. We use \Cref{Claim: Multi-Knapsack IBE} to show that a solution~$\chi'$ is a better solution of~$I$ if and only inf~$\val(\chi') > \val(\chi)$.

$(\Rightarrow)$ If~$\chi'$ is a better solution than~$\chi$, the~$(m+1)$-partition~$\chi'$ also fits into the knapsacks and thus, we have~$\val(\chi') = \score (\chi') > \score(\chi) = \val(\chi)$.

$(\Leftarrow)$ Conversely, let~$\chi'$ be an~$m$-partition with~$\val(\chi') > \val(\chi)$. Then~$\val(\chi') \geq 0$ and thus,~$\chi'$ fits into the knapsacks. Consequently, we have~$\val(\chi') = \score (\chi')$ and thus,~$\chi'$ is a better solution than~$\chi$.
\end{proof}

In \Cref{MKisHard} we complement the above algorithm by showing that \MKLoc is W[1]-hard parameterized by~$k$ and that an algorithm with running time~$n^{o(k)}$ contradicts the ETH.

\subsection{Vertex Deletion Distance to Specific Graph Properties}
\label{sec:VDD}

We next consider a general class of graph problems. 
Given a graph, one aims to delete a minimum number of vertices such that the remaining graph satisfies a specific property (for example: being bipartite). For any fixed graph-property~$\Pi$ that can be verified in polynomial time, we define the following problem.

\defproblem{\PiDel}{A graph~$G=(V,E)$.}{Find a subset~$S \subseteq V$ such that~$G - S$ fulfills property~$\Pi$ and~$|S|$ is minimal.}

Parameterized local search was considered for~\PiDel for general graph properties~$\Pi$~\cite{FFLRSV12} and in particular for the special case of~$\Pi$ being the family of all edgeless graphs, that is, for~\VC~\cite{FFLRSV12,GKO+12,KK17,KM22}. 
For the corresponding local search problems~\LSPiDel and~\LVC, the considered local neighborhood is the~\emph{$k$-swap neighborhood}, that is, the set of all solutions that have a symmetric difference with the current solution of size at most~$k$.
This neighborhood coincides with the~$k$-flip neighborhood if a solution for~\PiDel and~\VC is represented as a~$2$-coloring~$\chi$ of the vertex set~$V$, where~$G[\chi^{-1}(1)]$ fulfills property~$\Pi$ and the goal is to minimize the number of vertices that receive color~$2$ under~$\chi$.

It was shown that \LSPiDel is~\W-hard when parameterized by~$k$ for all hereditary graph properties~$\Pi$ that contain all edgeless graphs but not all cliques or vice versa~\cite{FFLRSV12}.
This holds in particular for the special case of~\LVC~\cite{FFLRSV12,GKO+12,KM22}. 
Furthermore, a closer inspection shows that \LVC cannot be solved in $f(k)\cdot n^{o(k)}$~time unless the ETH fails~\cite[Theorem~3.7]{KM22}.
Since \LVC is a special case of \LSPiDel, we also obtain the same hardness results for the more general \LSPiDel.
From the positive side, \LVC admits an FPT-algorithm for~$k$, if the input graph has a bounded local treewidth~\cite{FFLRSV12}. 
Moreover, algorithms for~\LVC are known that run in $f(k) \cdot \ell^{\Oh(k)} \cdot n^{\Oh(1)}$~time, where~$\ell$ is any of (i)~the maximum degree~$\Delta$~\cite{KK17}, (ii)~the~$h$-index of the input graph~\cite{KM22}, or (iii)~the treewidth of the input graph~\cite{KM22}.
In particular, the performance of the algorithm combining~$k$ and~$\Delta$ was shown to be very successful~\cite{KK17}.

In the following, we show that \LSPiDel can be solved in~$\nd^k \cdot 2^{\Oh(k)} \cdot n^{\Oh(1)}$~time and in~$k^{\Oh(\nd)} \cdot n^{\Oh(1)}$~time.
To this end, we show that our framework is applicable for an even more general problem in which we aim to find a vertex set of minimal size to remove, so that the remaining vertices can be partitioned into~$c$ sets, that each fulfill property~$\Pi$.
Formally, the problem is defined as follows for polynomial-time checkable graph properties~$\Pi$.

\defproblem{\MCPi}{A graph~$G=(V,E)$ and an integer~$c$.}{Find a~$c+1$-partition~$\chi$ of~$V$ such that~$G[\chi^{-1}(i)]$ has property~$\Pi$ for each~$i \in [1,c]$, and~$|\chi^{-1}(c+1)|$ is minimal.}

\MCPi generalizes~\PiDel and also other well-known graph problems.
For example, if~$\Pi$ is the family of complete graphs, \MCPi corresponds to the vertex-deletion distance to obtain a graph having a~$c$-clique-cover.

The input of the corresponding local search problem~\LSMCPi additionally consists of a~$c+1$-partition~$\chi$ of~$V$ such that~$G[\chi^{-1}(i)]$ has property~$\Pi$ for each~$i \in [1,c]$ together with some~$k \in \mathds{N}$ and one asks for a~$c+1$-partition~$\chi'$ with~$\dflip(\chi, \chi') \leq k$ and~$|\chi'^{-1}(c+1)|<|\chi^{-1}(c+1)|$. 

We now show the following.

\begin{theorem} \label{Theorem: Algo Vertex del to Pi}
\LSMCPi can be solved in~$\nd^k \cdot 2^{\Oh(k)} \cdot n^{\Oh(1)}$~time and in~$k^{\Oh(\nd)} \cdot n^{\Oh(1)}$~time.
\end{theorem}

\begin{proof}
The proof relies on the algorithm behind \Cref{Theorem: Framework Algorithms}. We describe how to use this algorithm to solve an instance~$I := (G=(V,E), c, \chi, k)$ of \LSMCPi. 

\textit{1. GBP Construction:} We describe how~$I$ corresponds to an instance~$J:=(b,X,(\varphi_i)_{i \in [1,b]},f,k)$ of \LSGBP. Our universe~$X$ is exactly the vertex set~$V$ of the input graph~$G$, and the number of bins~$b$ is~$c+1$. It remains to define the IBE~$(\varphi_i)_{i \in [1,b]}$. Intuitively, $\phi_i$ for~$i \in [1,c]$ essentially check whether~$G[\chi^{-1}(i)]$ has property~$\Pi$ by assigning~$\infty$ if this is not the case. For~$i \in [1,c]$ we thus define
\begin{align*}
\varphi_i(X') := 
\begin{cases}
\infty &\text{ if }G[X']\text{ satisfies property }\Pi\text{, or}\\
0	&\text{otherwise}.
\end{cases}
\end{align*}
Moreover, we define~$\varphi_{c+1} (X'):=|X'|$ for every~$X' \subseteq X$. Since we only consider properties~$\Pi$ that can be checked in polynomial time, one value~$\varphi(X')$ can be computed in polynomial time in~$|I|$, where~$I$ denotes the total encoding size of the instance~$I$. Our operation~$\oplus$ is the sum of integer numbers.

\textit{2. Type Partition:} Let~$(X_1, \dots, X_{\nd})$ be the collection of neighborhood classes of the input graph~$G$, which can be computed in~$\Oh(n+m)$~time~\cite{HM91}. Two vertices from the same neighborhood class are obviously target equivalent according to \Cref{Definition: Types} as replacing any vertex by a vertex with the same neighborhood results in the same induced subgraph structure.

\textit{3. Solution Correspondence:} Note that the set of solutions~$\chi'$ of \LSMCPi is exactly the set of all~$c+1$-partitions of~$V$. Since~$X=V$ and~$b=c+1$, the solutions of~$J$ have a one-to-one correspondence to the solutions of~$I$. Note that we have the following correspondence between the value from \Cref{Definition: Target Function} and the value~$|\chi^{-1}(c+1)|$.

\begin{claim} \label{Claim: Pi Deletion IBE}
Let~$\chi$ be a~$c+1$-partition of~$V$. Then,~$G[\chi^{-1}(i)]$ satisfies~$\Pi$ for all~$i \in [1,c]$ if and only if~$\val(\chi) \neq \infty$. Moreover, if~$G[\chi^{-1}(i)]$ satisfies~$\Pi$ for all~$i \in [1,c]$, then~$\val(\chi) = |\chi^{-1}(c+1)|$.
\end{claim}

Next, let~$\chi$ be the input solution of the instance~$I$. We use \Cref{Claim: Pi Deletion IBE} to show that a solution~$\chi'$ is a better solution of~$I$ if and only if~$\val(\chi') < \val(\chi)$. 

$(\Rightarrow)$ If~$\chi'$ is better than~$\chi$, all induced subgraphs~$G[\chi'^{-1}(i)]$ satisfy~$\Pi$ for~$i \in [1,c]$ and~$|\chi'^{-1}(c+1)| < |\chi^{-1}(c+1)|$. Thus,~$\val(\chi') < \val(\chi)$.

$(\Leftarrow)$ Conversely, let~$\chi'$ be a~$c+1$-partition with~$\val(\chi') < \val(\chi)$. Then~$\val(\chi') \neq \infty$ and thus~$G[\chi'^{-1}(i)]$ satisfies~$\Pi$ for~$i \in [1,c]$. Consequently, we have~$\val(\chi') = |\chi'^{-1}(c+1)|$ and thus,~$\chi'$ is a better solution than~$\chi$.

\end{proof}

This implies the following for the most-frequently considered parameterized local search problems~\LVC and~\LSPiDel.

\begin{corollary} 
\LSPiDel and~\LVC can be solved in~$\nd^k \cdot 2^{\Oh(k)} \cdot n^{\Oh(1)}$~time and in~$k^{\Oh(\nd)} \cdot n^{\Oh(1)}$~time.
\end{corollary}

\section{Intractability Results for the Considered Local Search Problems}
\label{sec:hardness}
In this section, we present (parameterized) intractability results for the considered local search problems in this work and tight running-time lower bounds with respect to the $\tau^k \cdot 2^{\Oh(k)}\cdot |I|^{\Oh(1)}$-time algorithms derived in~\Cref{Section: Framework Application}.
As already discussed in the previous sections, some of the considered problems in this work are known to be~\W-hard when parameterized by~$k$ and cannot be solved in $f(k) \cdot |I|^{o(k)}$~time, unless the ETH fails.
This is the case for~\LMC~\cite{M24}, \LSCE~\cite{GMNW23}, \LSPiDel~\cite{FFLRSV12,GKO+12,KM22}.
In this section we thus only show that these intractability results also hold for the remaining local search problems considered in this work, that is, for \MKLoc, \NSWLoc, and \BinLoc.

\subsection{Multi Knapsack and Nash Social Welfare}

\label{sec:Hardness NSW}
\newcommand{\dSum}{\textsc{$d$-Sum}\xspace}
\newcommand{\ledSum}{\textsc{Positive $d$-Sum}\xspace}

We start by presenting our (parameterized) intractability results and running-time lower bounds for the local search problems~\MKLoc and~\NSWLoc.
To this end, we reduce in both cases from the following problem.

\defdecproblem{\ledSum}{A set~$S$ of~$n$ positive (binary encoded) integers, a target value~$t$, and~$d\in \mathds{N}$.}{Is there a subset~$S'\subseteq S$ of size exactly~$d$, such that~$\sum_{s\in S'} s = t$.}

In \ledSum, all integers have a binary encoding. To derive the desired intractability results, we first show that~\ledSum provides these intractability results even on restricted instances.
Let~$I:=(S,d,t)$ be an instance of~\ledSum.
We say that~$I$ is~\emph{size-limiting} if for each subset~$S^*\subseteq S$ with~$|S^*|\neq d$, $\sum_{s\in S^*} s \neq t$.

\begin{lemma}\label{ledSum is hard}
Even on size-limiting instances, \ledSum is \W-hard with respect to~$d$ and cannot be solved in $n^{o(d)}$~time, unless the ETH fails.
\end{lemma}
\begin{proof}
We reduce from \dSum.

\defdecproblem{\dSum}{A set~$S$ of~$n$ integers with and~$d\in \mathds{N}$.}{Is there a subset~$S'\subseteq S$ of size exactly~$d$, such that~$\sum_{s\in S'} s = 0$.}

In \dSum, all integers have a binary encoding. It is known that~\dSum is~\W-hard when parameterized by~$d$~\cite{ALW14} and cannot be solved in $|S|^{o(d)}$~time, unless the ETH fails~\cite{PW10}.
Let~$I:=(S:= \{s_1, \dots, s_n\},d)$ be an instance of~\dSum and let~$\alpha := \max_{s\in S} |s|$.
For each~$i\in [1,n]$, we set~$r_i := s_i + (2\cdot d + 1)\cdot \alpha$.
Additionally, we set~$R:= \{r_i\colon 1 \leq i \leq n\}$ and~$t = (2\cdot d + 1) \cdot d \cdot \alpha$ and show that~$I$ is a yes-instance of~\dSum if and only if~$I':=(R,t,d' := d)$ is a yes-instance of~\ledSum.
This follows by the fact that for each index set~$J \subseteq [1,n]$, $\sum_{i\in J} r_i = (\sum_{i\in J} s_i) + |J| \cdot (2\cdot d + 1)\cdot \alpha$.
That is, if~$J$ has size exactly~$d$, then~$\sum_{i\in J} s_i = 0$ if and only if~$\sum_{i\in J} r_i = |J| \cdot (2\cdot d + 1)\cdot \alpha = t$.
Hence, $I$ is a yes-instance of~\dSum if~$I'$ is a yes-instance of~\ledSum.

Next we show that~$I'$ is size-limiting.
To this end, note that each number in~$R$ is between~$2\cdot d \cdot \alpha$ and~$2\cdot d \cdot \alpha + 2\cdot \alpha$.
Hence, each set~$R' \subseteq R$ of size less than~$d$ fulfills~$\sum_{r_i\in R'} r_i \leq |R'| \cdot (2\cdot d \cdot \alpha + 2\cdot \alpha) \leq d \cdot 2\cdot d \cdot \alpha  < d \cdot (2\cdot d + 1) \cdot \alpha = t$.
This implies that no subset of~$R'$ of size less than~$d$ can sum up to~$t$.
Similarly, each set~$R' \subseteq R$ of size larger than~$d$ fulfills~$\sum_{r_i\in R'} r_i \geq |R'| \cdot (2\cdot d \cdot \alpha) \geq (d+1) \cdot 2\cdot d \cdot \alpha = t + d \cdot \alpha > t$, since~$\alpha$ is non-zero in all reasonable instances of~\dSum.
This implies that no subset of~$R'$ of size larger than~$d$ can sum up to~$t$.
In total, only subsets of~$R$ of size exactly~$d$ can sum up to~$t$.

Recall that \dSum is~\W-hard when parameterized by~$d$~\cite{ALW14} and cannot be solved in $|S|^{o(d)}$~time, unless the ETH fails~\cite{PW10}.
Since~$|S| = |R|$ and~$d'=d$, this implies that even on size-limiting instances, \ledSum is~\W-hard when parameterized by~$d$ and cannot be solved in $|R|^{o(d)}$~time, unless the ETH fails.
\end{proof}

Based on this lemma, we are now ready to show our intractability results for~\MKLoc.

\begin{theorem}\label{MKisHard}
\MKLoc (with binary encoded weights and values) is \W-hard with respect to~$k$ and cannot be solved in $n^{o(k)}$~time, unless the ETH fails, where~$n$ denotes the number of items in the input instance.
This holds even if 
\begin{itemize}
\item there is only a single knapsack,
\item the weight of each item is equal to its value, and
\item the initial solution is locally optimal if and only if it is globally optimal.
\end{itemize} 
\end{theorem}
\begin{proof}
We reduce from~\ledSum on size-limiting instances.
Let~$I:=(S,d,t)$ be a size-limiting of~\ledSum.
Moreover, let~$t^*$ denote the sum of all numbers in~$S$.
We define a set of items~$X$ as follows: for each number~$s_i\in S$, we add an item~$x_i$ to~$X$ with both value and weight equal to~$s_i$.
Additionally, we add an item~$x^*$ to~$X$ of weight and value equal to~$t+1$.
Finally, we define a unique knapsack of capacity~$t^*+1$ which contains exactly the items from~$X\setminus \{x^*\}$ in the initial solution~$\chi$.
Let~$I':=(t^*+1,X,\chi,k)$ be the resulting instance of~\MKLoc, where~$k:=d+1$.

We show that~$I$ is a yes-instance of~\ledSum if and only if~$I'$ is a yes-instance of~\MKLoc.
To this end, note that 
\begin{enumerate}
\item[(i)] the initial solution has value~$t^*$,
\item[(ii)] each solution has value at most~$t^*+1$ (since the value of each item is equal to its weight and the knapsack can only be filled with weight at most~$t^*+1$), and
\item[(iii)] each solution with value larger than~$t^*$ has to contain item~$x^*$.
\end{enumerate}
Since the weight and value of item~$x^*$ are~$t+1$, this implies that each improving flip (of arbitrary size) has to (a)~flip~$x^*$ into the knapsack and (b)~flip items out of the knapsack for which the sum of weights (sum of values) is equal to~$t$.
Since the set of weights of the items in~$X\setminus \{x^*\}$ is exactly the set of numbers~$S$, there is a subset~$S'$ of of numbers in~$S$ that sum up to exactly~$t$ if and only if such an improving flip of size~$|S|+1$ exists.
Since~$I$ is size-limiting, only subsets of~$S$ of size exactly~$d$ may sum up to~$t$.
This implies that the following statements are equivalent:
\begin{itemize}
\item $I$ is a yes-instance of~\ledSum, 
\item $I'$ is a yes-instance of~\MKLoc, and
\item the initial solution for the~\MKLoc-instance~$I'$ is not globally optimal.
\end{itemize}

Due to~\Cref{ledSum is hard}, even on size-limiting instances, \ledSum is \W-hard with respect to~$d$ and cannot be solved in $|S|^{o(d)}$~time, unless the ETH fails.
Since~$|X| = |S|+1$ and~$k = d+1$, this implies that under the stated restrictions \MKLoc is \W-hard with respect to~$k$ and cannot be solved in $n^{o(k)}$~time, unless the ETH fails.
\end{proof}

With a similar reduction from~\ledSum, we can also derive our intractability results for~\NSWLoc.

\begin{theorem}\label{NSWisHard}
\NSWLoc (with binary encoded evaluations) is \W-hard with respect to~$k$ and cannot be solved in $n^{o(k)}$~time, unless the ETH fails.
This holds even if  
\begin{itemize}
\item there are only two agents,
\item each item is evaluated equally by both agents, and
\item the initial solution is locally optimal if and only if it is globally optimal.
\end{itemize} 
\end{theorem}
\begin{proof}
We reduce from~\ledSum on size-limiting instances.
Let~$I:=(S,d,t)$ be a size-limiting instance of~\ledSum.
Moreover, let~$t^*$ denote the sum of all numbers in~$S$.
Note that the above restriction implies that only subsets of~$S$ of size exactly~$|S|-d$ may sum up to~$t^*-t$.
Hence, we can assume that no subset of size at most~$d$ can sum up to~$t^*-t$, as otherwise~$|S|\leq 2\cdot d$ and we could solve the instance~$I$ of~\ledSum in~$f(d)$~time.

We define a set of items~$X$ as follows: for each number~$s_i\in S$, we add an item~$x_i$ to~$X$ which gives both agents a utility of~$s_i$.
Additionally, we add two further items~$x^*$ and~$y^*$ to~$X$.
The item~$x^*$ gives both agents a utility of~$t+1$ and the item~$y^*$ gives both agents a utility of~$t^*-t+1$.
Finally, we set~$k:= d+1$ and define an initial allocation~$\chi$ that assigns (i)~all items of~$X\setminus \{x^*,y^*\}$ to the first agent and (ii)~both~$x^*$ and~$y^*$ to the second agent.
Let~$I'$ be the resulting instance of~\NSWLoc.
In the following, we show that~$I$ is a yes-instance of~\ledSum if and only if~$I'$ is a yes-instance of~\NSWLoc.

Note that in the initial solution~$\chi$, the first agent is assigned a total utility of~$t^*$ and the second agent is assigned a total utility of~$t^*+2$.
Hence, the Nash score of the initial solution~$\chi$ is~$t^* \cdot (t^*+2)$.
This implies that in every solution~$\chi'$ that improves over~$\chi$, each agent is assigned a total utility of exactly~$t^*+1$.
Consequently, in each such improving solution~$\chi'$, $x^*$ and~$y^*$ are not assigned to the same agent, since their total utility exceeds~$t^* + 1$.
Hence, the flip between~$\chi$ and any such improving solution consist of either 
\begin{itemize}
\item flipping~$x^*$ from the second to the first agent and a set of items of total utility~$t$ from the first agent to the second agent, or
\item flipping~$y^*$ from the second to the first agent and a set of items of total utility~$t^*-t$ from the first agent to the second agent.
\end{itemize}
In other words, the initial solution~$\chi$ is \emph{not} globally optimal if and only if there is a set of items~$X'\subseteq X\setminus\{x^*,y^*\}$ of total utility either~$t$ or~$t^*-t$.
Since the total utility of~$X\setminus\{x^*,y^*\}$ is~$t^*$, this is the case if and only if there is a set of items~$X'\subseteq X\setminus\{x^*,y^*\}$ of total utility~$t$.
Recall that the set of utilities of the items in~$X'\subseteq X\setminus\{x^*,y^*\}$ is exactly the set of numbers~$S$.
Moreover, recall that (a)~$I$ is size-limiting, that is, only subsets of~$S$ of size exactly~$d = k-1$ may sum up to~$t$ and (b)~that no subset of~$S$ of size at most~$d$ can sum up to~$t^*-t$.
Hence, the following statements are equivalent:
\begin{itemize}
\item $I$ is a yes-instance of~\ledSum, 
\item $I'$ is a yes-instance of~\NSWLoc, and
\item the initial solution for the~\NSWLoc-instance~$I'$ is not globally optimal.
\end{itemize}

Due to~\Cref{ledSum is hard}, even on size-limiting instances, \ledSum is \W-hard with respect to~$d$ and cannot be solved in $|S|^{o(d)}$~time, unless the ETH fails.
Since~$|X| = n+2$ and~$k= d+1$, the above reduction thus implies that even under the stated restrictions, \NSWLoc is \W-hard with respect to~$k$ and cannot be solved in $|X|^{o(k)}$~time, unless the ETH fails.
\end{proof}

\subsection{Bin Packing}

Now, we provide matching hardness results for \BinLoc.
More precisely, we provide two hardness results even if each entry in each vector is only~0 or~1.
First, we show that \BinLoc is W[1]-hard with respect to~$k$ and that an algorithm with running time~$f(k) \cdot |I|^{o(k)}$ violates the ETH.
Consequently, our $\tau^{k}\cdot 2^{\Oh(k)}\cdot |I|^{\Oh(1)}$~time algorithm presented in \Cref{Theorem: Algo Vector Bin Packing} is tight if the ETH is true.
Second, we show W[1]-hardness for \BinLoc parameterized by~$k+q$. 
Here, $q:=\max_{v \in \SSS} \lVert v \rVert_1$ is the maximal sum of entries over all vectors.
Recall that~$q$ is usually smaller than~$\tau$ in our real-world application.

\begin{theorem}\label{prop-ls-bin-packing}
\BinLoc is \W-hard with respect to~$k$ and cannot be solved in $f(k) \cdot |I|^{o(k)}$~time, unless the ETH fails.
This holds even if 
\begin{itemize}
\item there are only two bins,
\item each entry in each vector is only~0 or~1,
\item each entry of the vector~$\omega$ for each bin is~1, and
\item the initial solution is locally optimal if and only if it is globally optimal.
\end{itemize}
\end{theorem}
\begin{proof}
We present a simple reduction from~\LMC which provides the desired intractability results.
Let~$c \geq 2$ and let~$I:=(G=(V,E),\chi: V \to [1,c],k)$ be an instance of~\LMC where~$\chi$ is locally optimal if and only if~$\chi$ is globally optimal.
Even under these restrictions, \LMC is \W-hard with respect to~$k$ and cannot be solved in $f(k) \cdot n^{o(k)}$~time, unless the ETH fails~\cite{GGKM23}.
Let~$m:= |E|$ and let~$E := \{e_1, \dots, e_m\}$.
We obtain an equivalent instance~$I':=(c,m,\SSS,(\omega_i)_{i\in [1,c]},\psi,k')$ of~\LVBP as follows:
For each bin~$j\in [1,c]$, we define the vector~$\omega_j$ as the vector of length~$m$ that has a 1 in each dimension.
For each vertex~$v\in V$, we create a vector~$x_v$ that has a 1 at dimension~$i\in [1,m]$ if and only if vertex~$v$ is incident with edge~$e_i$.
Let~$\SSS$ be the set of these vectors and let~$\psi$ be the coloring obtained from~$\chi$ by assigning for each vertex~$v\in V$, color~$\chi(v)$ to vector~$x_v$, that is, $\chi(v) = \psi(x_v)$.
Finally, we set~$k':=k$.

For the sake of simplicity, we may identify each vector~$x_v$ by its corresponding vertex~$v$.
Similarly, we may consider~$c$-partitions of~$V$ instead of~$c$-partitions of~$\SSS$ as solutions for~$I'$ based on the obvious bijection between vertices and vectors.
Next, we show that~$I$ is a yes-instance of~\LMC if and only if~$I'$ is a yes-instance of~\BinLoc.
To this end, we first analyze the objective functions of both instances with respect to corresponding solutions.

Let~$\chi'$ be a~$c$-coloring of~$V$ and let~$\psi'$ be the corresponding~$c$-coloring of~$\SSS$.
Let~$e_i:=\{u,v\}$ be an edge of~$G$ having endpoints of distinct color under~$\chi'$. 
Then, no bin produces an overload in dimension~$i$ because only the vectors~$x_u$ and~$x_v$ have a 1 at dimension~$i$ and both vectors receive distinct colors under~$\psi'$.
Similarly, let~$e_i:=\{u,v\}$ be an edge of~$G$ having endpoints of the same color~$\alpha\in [1,c]$ under~$\chi'$.
Then, no bin except~$\alpha$ produces an overload at dimension~$i$ and bin~$\alpha$ produces an overhead of exactly~$1$ in dimension~$i$, because only the vectors~$x_u$ and~$x_v$ have a 1 at dimension~$i$ and both vectors receive color~$\alpha$ under~$\psi'$.
Consequently, the total overload of~$\psi'$ over all dimensions and all bins equals~$|E|$ minus the number of properly colored edges of~$G$ under~$\chi'$.
In other words, $\chi'$ is a better solution for~$I$ than~$\chi$ if and only if~$\psi'$ is a better solution for~$I'$ than~$\psi$.
Since~$\dflip(\chi,\chi') = \dflip(\psi,\psi')$, this implies that~$I$ is a yes-instance of~\LMC if and only if~$I'$ is a yes-instance of~\BinLoc.
Furthermore, since~$\chi$ is a locally optimal solution if and only if~$\chi$ is a globally optimal solution, $\psi$ is a locally optimal solution if and only if~$\psi$ is a globally optimal solution.

Recall that \LMC is \W-hard with respect to~$k$ and cannot be solved in $f(k) \cdot n^{o(k)}$~time, unless the ETH fails~\cite{GGKM23}.
Since~$|\SSS| = n$, each vector has~$m \leq n^2$ dimensions, and~$k' = k$, this implies that \BinLoc is \W-hard with respect to~$k'$ and cannot be solved in $f(k') \cdot |I'|^{o(k')}$~time, unless the ETH fails.
\end{proof}

Now, we present our second hardness result for the parameter~$k$ plus~$q$, the maximal sum of entries over all vectors.

\begin{theorem}
\label{thm-w-hardness-ls-bin-packing}
\BinLoc is \W-hard with respect to~$k+q$, even if 
\begin{itemize}
\item each entry in each vector is only~0 or~1, and
\item each entry of the vector~$\omega$ for each bin is~1.
\end{itemize}
\end{theorem}
\begin{proof}
We reduce from the W[1]-hard \textsc{Multicolored Clique} problem~\cite{C+15}.
In \textsc{Multicolored Clique} the input is a graph~$G$ where~$V(G)$ is partitioned into $\ell$~sets~$(V_1,\ldots, V_\ell)$ and the task is to find a multicolored clique, that is, a clique containing exactly one vertex from each partite set.

\textbf{Intuition:}
The \BinLoc instance has three different types of bins:
First, for each vertex~$v_i\in V(G)$ we have a \emph{vertex bin~$B_{v_i}$}, second, for each edge~$\{u_i,w_j\}\in E(G)$ where~$u_i\in V_i$ and~$w_j\in V_j$ we have an \emph{edge bin~$B_{u_i,w_j}$}, and third we have a \emph{target bin~$B^*$}.
Furthermore, for each vertex~$v_i$, we have a \emph{vertex vector~$\vect(v_i)$}, and for each edge~$\{u_i,w_j\}$ we have an \emph{edge vector~$\vect(u_i,w_j)$}.
In the initial assignment~$f$ all vertex vectors are in the target bin~$B^*$ and each edge vector~$\vect(u_i,w_j)$ is in its corresponding edge bin~$B_{u_i,w_j}$.

We set~$k\coloneqq \ell + \binom{\ell}{2}$.
Let~$f'$ be another assignment of the vectors to the bins such that~$d_{\rm flip}(f,f') \le k$ which is improving upon~$f$.
$D_{\rm flip}(f,f')$ contains all $\ell$~vertex vectors corresponding to vertices of a multicolored clique~$C$ and all $\binom{\ell}{2}$~edge vectors corresponding to the edges of~$C$.
Also, each edge vector in~$D_{\rm flip}(f,f')$ is assigned to the target bin~$B^*$ by~$f'$ and each vertex vector~$\vect(v_i)$ in~$D_{\rm flip}(f,f')$ is assigned to its corresponding vertex bin~$B_{v_i}$ by~$f'$.

We achieve this as follows:
For each vertex~$v_i$ and each other color class~$j\in[1,\ell]\setminus\{i\}$ we create a large number~$z$ of \emph{important dimensions}.
Each vertex vector~$\vect(v_i)$ has value~1 in important dimensions with respect to~$(v_i,j)$ for all such~$j$.
Furthermore, each edge vector~$\vect(u_i,w_j)$ has value~1 in all important dimensions with respect to~$(u_i,j)$ and~$(w_j,i)$.
Hence, after the vertex vectors corresponding to a multicolored clique~$C$ are moved out of the target bin~$B^*$ the edge vectors corresponding to~$C$ can be moved into those empty positions in the target bin~$B^*$.
In the end, the overload decreases since the move of each vertex vector~$\vect(v_i)$ into its corresponding vertex bin~$B_{v_i}$ increases the overload by~$z/2-1$ while the move of each edge vector into the target bin~$B^*$ decreases the overload by~$z$.
By setting~$z$ accordingly, we can then show that~$C$ is a multicolored clique.

Finally, we use many \emph{dummy dimensions} and \emph{dummy and blocking vectors} to allow easy arguments on which vectors are contained in an improving $k$-flip and how they are reassigned.

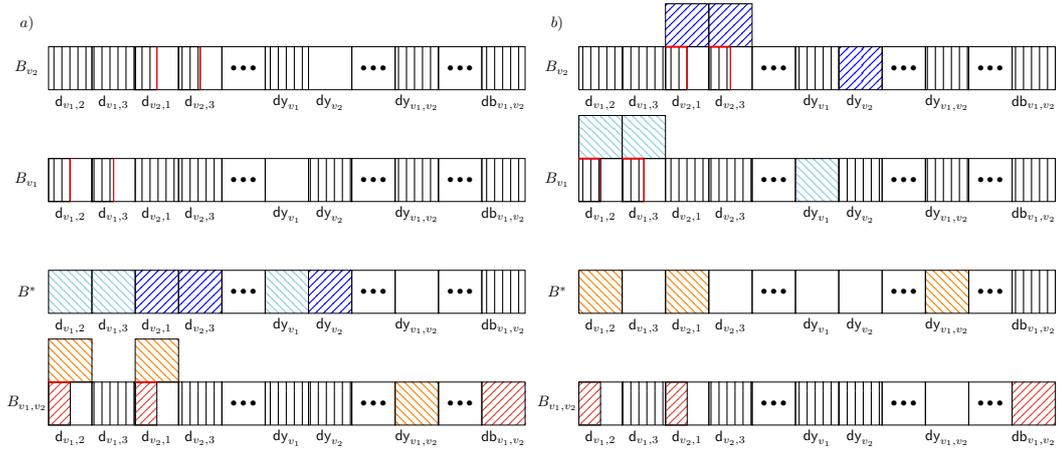
\begin{figure}[t]

\begin{minipage}{0.49\textwidth}
	\centering
	\scalebox{0.57}{
	\begin{tikzpicture}
\def\mydistx{1.0}
\def\mydisty{2.6}

\node[](d) at (-0.5*\mydistx, 0.5 + 2.4*\mydisty)  {$a)$};

\node[](d) at (-0.5*\mydistx, 0.5 + 0*\mydisty)  {$B^*$};
	
\draw[pattern=north west lines, pattern color=skyblue] (0*\mydistx, 0 + 0*\mydisty) rectangle (1*\mydistx, 1 + 0*\mydisty) {};
\node[](d) at (0.5*\mydistx, -0.3 + 0*\mydisty)  {$\di_{v_1,2}$};
\draw[pattern=north west lines, pattern color=skyblue] (1*\mydistx, 0 + 0*\mydisty) rectangle (2*\mydistx, 1 + 0*\mydisty) {};
\node[](d) at (1.5*\mydistx, -0.3 + 0*\mydisty)  {$\di_{v_1,3}$};

\draw[pattern=north east lines, pattern color=blue] (2*\mydistx, 0 + 0*\mydisty) rectangle (3*\mydistx, 1 + 0*\mydisty) {};
\node[](d) at (2.5*\mydistx, -0.3 + 0*\mydisty)  {$\di_{v_2,1}$};
\draw[pattern=north east lines, pattern color=blue] (3*\mydistx, 0 + 0*\mydisty) rectangle (4*\mydistx, 1 + 0*\mydisty) {};
\node[](d) at (3.5*\mydistx, -0.3 + 0*\mydisty)  {$\di_{v_2,3}$};

\draw[] (4*\mydistx, 0 + 0*\mydisty) rectangle (5*\mydistx, 1 + 0*\mydisty) {};
\node[](d) at (4.5*\mydistx, 0.5 + 0*\mydisty)  {\footnotesize $\bullet\bullet\bullet$};

\draw[pattern=north west lines, pattern color=skyblue] (5*\mydistx, 0 + 0*\mydisty) rectangle (6*\mydistx, 1 + 0*\mydisty) {};
\node[](d) at (5.5*\mydistx, -0.3 + 0*\mydisty)  {$\dy_{v_1}$};

\draw[pattern=north east lines, pattern color=blue] (6*\mydistx, 0 + 0*\mydisty) rectangle (7*\mydistx, 1 + 0*\mydisty) {};
\node[](d) at (6.5*\mydistx, -0.3 + 0*\mydisty)  {$\dy_{v_2}$};

\draw[] (7*\mydistx, 0 + 0*\mydisty) rectangle (8*\mydistx, 1 + 0*\mydisty) {};
\node[](d) at (7.5*\mydistx, 0.5 + 0*\mydisty)  {\footnotesize $\bullet\bullet\bullet$};

\draw[] (8*\mydistx, 0 + 0*\mydisty) rectangle (9*\mydistx, 1 + 0*\mydisty) {};
\node[](d) at (8.5*\mydistx, -0.3 + 0*\mydisty)  {$\dy_{v_1,v_2}$};

\draw[] (9*\mydistx, 0 + 0*\mydisty) rectangle (10*\mydistx, 1 + 0*\mydisty) {};
\node[](d) at (9.5*\mydistx, 0.5 + 0*\mydisty)  {\footnotesize $\bullet\bullet\bullet$};     

\draw[pattern=vertical lines] (10*\mydistx, 0 + 0*\mydisty) rectangle (11*\mydistx, 1 + 0*\mydisty) {};
\node[](d) at (10.5*\mydistx, -0.3 + 0*\mydisty)  {$\db_{v_1,v_2}$};

\node[](d) at (-0.5*\mydistx, 0.5 + 1*\mydisty)  {$B_{v_1}$};
	
\draw[] (0*\mydistx, 0 + 1*\mydisty) rectangle (1*\mydistx, 1 + 1*\mydisty) {};
\draw[pattern=vertical lines] (0*\mydistx, 0 + 1*\mydisty) rectangle (0.5*\mydistx, 1 + 1*\mydisty) {};
\draw[red, line width=3pt] (0.5*\mydistx, 0 + 1*\mydisty) rectangle (0.5*\mydistx, 1 + 1*\mydisty) {};
\node[](d) at (0.5*\mydistx, -0.3 + 1*\mydisty)  {$\di_{v_1,2}$};
\draw[] (1*\mydistx, 0 + 1*\mydisty) rectangle (2*\mydistx, 1 + 1*\mydisty) {};
\draw[pattern=vertical lines] (1*\mydistx, 0 + 1*\mydisty) rectangle (1.5*\mydistx, 1 + 1*\mydisty) {};
\draw[red, line width=3pt] (1.5*\mydistx, 0 + 1*\mydisty) rectangle (1.5*\mydistx, 1 + 1*\mydisty) {};
\node[](d) at (1.5*\mydistx, -0.3 + 1*\mydisty)  {$\di_{v_1,3}$};

\draw[pattern=vertical lines] (2*\mydistx, 0 + 1*\mydisty) rectangle (3*\mydistx, 1 + 1*\mydisty) {};
\node[](d) at (2.5*\mydistx, -0.3 + 1*\mydisty)  {$\di_{v_2,1}$};
\draw[pattern=vertical lines] (3*\mydistx, 0 + 1*\mydisty) rectangle (4*\mydistx, 1 + 1*\mydisty) {};
\node[](d) at (3.5*\mydistx, -0.3 + 1*\mydisty)  {$\di_{v_2,3}$};

\draw[] (4*\mydistx, 0 + 1*\mydisty) rectangle (5*\mydistx, 1 + 1*\mydisty) {};
\node[](d) at (4.5*\mydistx, 0.5 + 1*\mydisty)  {\footnotesize $\bullet\bullet\bullet$};

\draw[] (5*\mydistx, 0 + 1*\mydisty) rectangle (6*\mydistx, 1 + 1*\mydisty) {};
\node[](d) at (5.5*\mydistx, -0.3 + 1*\mydisty)  {$\dy_{v_1}$};

\draw[pattern=vertical lines] (6*\mydistx, 0 + 1*\mydisty) rectangle (7*\mydistx, 1 + 1*\mydisty) {};
\node[](d) at (6.5*\mydistx, -0.3 + 1*\mydisty)  {$\dy_{v_2}$};

\draw[] (7*\mydistx, 0 + 1*\mydisty) rectangle (8*\mydistx, 1 + 1*\mydisty) {};
\node[](d) at (7.5*\mydistx, 0.5 + 1*\mydisty)  {\footnotesize $\bullet\bullet\bullet$};

\draw[pattern=vertical lines] (8*\mydistx, 0 + 1*\mydisty) rectangle (9*\mydistx, 1 + 1*\mydisty) {};
\node[](d) at (8.5*\mydistx, -0.3 + 1*\mydisty)  {$\dy_{v_1,v_2}$};

\draw[] (9*\mydistx, 0 + 1*\mydisty) rectangle (10*\mydistx, 1 + 1*\mydisty) {};
\node[](d) at (9.5*\mydistx, 0.5 + 1*\mydisty)  {\footnotesize $\bullet\bullet\bullet$};  

\draw[pattern=vertical lines] (10*\mydistx, 0 + 1*\mydisty) rectangle (11*\mydistx, 1 + 1*\mydisty) {};
\node[](d) at (10.5*\mydistx, -0.3 + 1*\mydisty)  {$\db_{v_1,v_2}$};

\draw[white] (0*\mydistx, 1 + 2*\mydisty) rectangle (1*\mydistx, 2 + 2*\mydisty) {};

\node[](d) at (-0.5*\mydistx, 0.5 + 2*\mydisty)  {$B_{v_2}$};
	
\draw[pattern=vertical lines] (0*\mydistx, 0 + 2*\mydisty) rectangle (1*\mydistx, 1 + 2*\mydisty) {};
\node[](d) at (0.5*\mydistx, -0.3 + 2*\mydisty)  {$\di_{v_1,2}$};
\draw[pattern=vertical lines] (1*\mydistx, 0 + 2*\mydisty) rectangle (2*\mydistx, 1 + 2*\mydisty) {};
\node[](d) at (1.5*\mydistx, -0.3 + 2*\mydisty)  {$\di_{v_1,3}$};

\draw[] (2*\mydistx, 0 + 2*\mydisty) rectangle (3*\mydistx, 1 + 2*\mydisty) {};
\draw[pattern=vertical lines] (2*\mydistx, 0 + 2*\mydisty) rectangle (2.5*\mydistx, 1 + 2*\mydisty) {};
\draw[red, line width=3pt] (2.5*\mydistx, 0 + 2*\mydisty) rectangle (2.5*\mydistx, 1 + 2*\mydisty) {};
\node[](d) at (2.5*\mydistx, -0.3 + 2*\mydisty)  {$\di_{v_2,1}$};
\draw[] (3*\mydistx, 0 + 2*\mydisty) rectangle (4*\mydistx, 1 + 2*\mydisty) {};
\draw[pattern=vertical lines] (3*\mydistx, 0 + 2*\mydisty) rectangle (3.5*\mydistx, 1 + 2*\mydisty) {};
\draw[red, line width=3pt] (3.5*\mydistx, 0 + 2*\mydisty) rectangle (3.5*\mydistx, 1 + 2*\mydisty) {};
\node[](d) at (3.5*\mydistx, -0.3 + 2*\mydisty)  {$\di_{v_2,3}$};

\draw[] (4*\mydistx, 0 + 2*\mydisty) rectangle (5*\mydistx, 1 + 2*\mydisty) {};
\node[](d) at (4.5*\mydistx, 0.5 + 2*\mydisty)  {\footnotesize $\bullet\bullet\bullet$};

\draw[pattern=vertical lines] (5*\mydistx, 0 + 2*\mydisty) rectangle (6*\mydistx, 1 + 2*\mydisty) {};
\node[](d) at (5.5*\mydistx, -0.3 + 2*\mydisty)  {$\dy_{v_1}$};

\draw[] (6*\mydistx, 0 + 2*\mydisty) rectangle (7*\mydistx, 1 + 2*\mydisty) {};
\node[](d) at (6.5*\mydistx, -0.3 + 2*\mydisty)  {$\dy_{v_2}$};

\draw[] (7*\mydistx, 0 + 2*\mydisty) rectangle (8*\mydistx, 1 + 2*\mydisty) {};
\node[](d) at (7.5*\mydistx, 0.5 + 2*\mydisty)  {\footnotesize $\bullet\bullet\bullet$};

\draw[pattern=vertical lines] (8*\mydistx, 0 + 2*\mydisty) rectangle (9*\mydistx, 1 + 2*\mydisty) {};
\node[](d) at (8.5*\mydistx, -0.3 + 2*\mydisty)  {$\dy_{v_1,v_2}$};

\draw[] (9*\mydistx, 0 + 2*\mydisty) rectangle (10*\mydistx, 1 + 2*\mydisty) {};
\node[](d) at (9.5*\mydistx, 0.5 + 2*\mydisty)  {\footnotesize $\bullet\bullet\bullet$}; 

\draw[pattern=vertical lines] (10*\mydistx, 0 + 2*\mydisty) rectangle (11*\mydistx, 1 + 2*\mydisty) {};
\node[](d) at (10.5*\mydistx, -0.3 + 2*\mydisty)  {$\db_{v_1,v_2}$};

\node[](d) at (-0.5*\mydistx, 0.5 + -1*\mydisty)  {$B_{v_1,v_2}$};
	
\draw[] (0*\mydistx, 0 + -1*\mydisty) rectangle (1*\mydistx, 1 + -1*\mydisty) {};
\draw[pattern=north east lines, pattern color=vermilion] (0*\mydistx, 0 + -1*\mydisty) rectangle (0.5*\mydistx, 1 + -1*\mydisty) {};
\node[](d) at (0.5*\mydistx, -0.3 + -1*\mydisty)  {$\di_{v_1,2}$};
\draw[pattern=north west lines, pattern color=orange] (0*\mydistx, 1 + -1*\mydisty) rectangle (1*\mydistx, 2 + -1*\mydisty) {};
\draw[red, line width=3pt] (0*\mydistx, 1 + -1*\mydisty) rectangle (0.5*\mydistx, 1 + -1*\mydisty) {};

\draw[pattern=vertical lines] (1*\mydistx, 0 + -1*\mydisty) rectangle (2*\mydistx, 1 + -1*\mydisty) {};
\node[](d) at (1.5*\mydistx, -0.3 + -1*\mydisty)  {$\di_{v_1,3}$};

\draw[] (2*\mydistx, 0 + -1*\mydisty) rectangle (3*\mydistx, 1 + -1*\mydisty) {};
\node[](d) at (2.5*\mydistx, -0.3 + -1*\mydisty)  {$\di_{v_2,1}$};
\draw[pattern=north east lines, pattern color=vermilion] (2*\mydistx, 0 + -1*\mydisty) rectangle (2.5*\mydistx, 1 + -1*\mydisty) {};
\draw[pattern=north west lines, pattern color=orange] (2*\mydistx, 1 + -1*\mydisty) rectangle (3*\mydistx, 2 + -1*\mydisty) {};
\draw[red, line width=3pt] (2*\mydistx, 1 + -1*\mydisty) rectangle (2.5*\mydistx, 1 + -1*\mydisty) {};

\draw[pattern=vertical lines] (3*\mydistx, 0 + -1*\mydisty) rectangle (4*\mydistx, 1 + -1*\mydisty) {};
\node[](d) at (3.5*\mydistx, -0.3 + -1*\mydisty)  {$\di_{v_2,3}$};

\draw[] (4*\mydistx, 0 + -1*\mydisty) rectangle (5*\mydistx, 1 + -1*\mydisty) {};
\node[](d) at (4.5*\mydistx, 0.5 + -1*\mydisty)  {\footnotesize $\bullet\bullet\bullet$};

\draw[pattern=vertical lines] (5*\mydistx, 0 + -1*\mydisty) rectangle (6*\mydistx, 1 + -1*\mydisty) {};
\node[](d) at (5.5*\mydistx, -0.3 + -1*\mydisty)  {$\dy_{v_1}$};

\draw[pattern=vertical lines] (6*\mydistx, 0 + -1*\mydisty) rectangle (7*\mydistx, 1 + -1*\mydisty) {};
\node[](d) at (6.5*\mydistx, -0.3 + -1*\mydisty)  {$\dy_{v_2}$};

\draw[] (7*\mydistx, 0 + -1*\mydisty) rectangle (8*\mydistx, 1 + -1*\mydisty) {};
\node[](d) at (7.5*\mydistx, 0.5 + -1*\mydisty)  {\footnotesize $\bullet\bullet\bullet$};

\draw[pattern=north west lines, pattern color=orange] (8*\mydistx, 0 + -1*\mydisty) rectangle (9*\mydistx, 1 + -1*\mydisty) {};
\node[](d) at (8.5*\mydistx, -0.3 + -1*\mydisty)  {$\dy_{v_1,v_2}$};

\draw[] (9*\mydistx, 0 + -1*\mydisty) rectangle (10*\mydistx, 1 + -1*\mydisty) {};
\node[](d) at (9.5*\mydistx, 0.5 + -1*\mydisty)  {\footnotesize $\bullet\bullet\bullet$}; 

\draw[pattern=north east lines, pattern color=vermilion] (10*\mydistx, 0 + -1*\mydisty) rectangle (11*\mydistx, 1 + -1*\mydisty) {};
\node[](d) at (10.5*\mydistx, -0.3 + -1*\mydisty)  {$\db_{v_1,v_2}$};
     	
\end{tikzpicture}
}
\end{minipage}
\begin{minipage}{0.49\textwidth}
\centering
	\scalebox{0.57}{
	\begin{tikzpicture}
\def\mydistx{1.0}
\def\mydisty{2.6}

\node[](d) at (-0.5*\mydistx, 0.5 + 2.4*\mydisty)  {$b)$};

\node[](d) at (-0.5*\mydistx, 0.5 + 0*\mydisty)  {$B^*$};
	
\draw[pattern=north west lines, pattern color=orange] (0*\mydistx, 0 + 0*\mydisty) rectangle (1*\mydistx, 1 + 0*\mydisty) {};
\node[](d) at (0.5*\mydistx, -0.3 + 0*\mydisty)  {$\di_{v_1,2}$};
\draw[] (1*\mydistx, 0 + 0*\mydisty) rectangle (2*\mydistx, 1 + 0*\mydisty) {};
\node[](d) at (1.5*\mydistx, -0.3 + 0*\mydisty)  {$\di_{v_1,3}$};

\draw[pattern=north west lines, pattern color=orange] (2*\mydistx, 0 + 0*\mydisty) rectangle (3*\mydistx, 1 + 0*\mydisty) {};
\node[](d) at (2.5*\mydistx, -0.3 + 0*\mydisty)  {$\di_{v_2,1}$};
\draw[] (3*\mydistx, 0 + 0*\mydisty) rectangle (4*\mydistx, 1 + 0*\mydisty) {};
\node[](d) at (3.5*\mydistx, -0.3 + 0*\mydisty)  {$\di_{v_2,3}$};

\draw[] (4*\mydistx, 0 + 0*\mydisty) rectangle (5*\mydistx, 1 + 0*\mydisty) {};
\node[](d) at (4.5*\mydistx, 0.5 + 0*\mydisty)  {\footnotesize $\bullet\bullet\bullet$};

\draw[] (5*\mydistx, 0 + 0*\mydisty) rectangle (6*\mydistx, 1 + 0*\mydisty) {};
\node[](d) at (5.5*\mydistx, -0.3 + 0*\mydisty)  {$\dy_{v_1}$};

\draw[] (6*\mydistx, 0 + 0*\mydisty) rectangle (7*\mydistx, 1 + 0*\mydisty) {};
\node[](d) at (6.5*\mydistx, -0.3 + 0*\mydisty)  {$\dy_{v_2}$};

\draw[] (7*\mydistx, 0 + 0*\mydisty) rectangle (8*\mydistx, 1 + 0*\mydisty) {};
\node[](d) at (7.5*\mydistx, 0.5 + 0*\mydisty)  {\footnotesize $\bullet\bullet\bullet$};

\draw[pattern=north west lines, pattern color=orange] (8*\mydistx, 0 + 0*\mydisty) rectangle (9*\mydistx, 1 + 0*\mydisty) {};
\node[](d) at (8.5*\mydistx, -0.3 + 0*\mydisty)  {$\dy_{v_1,v_2}$};

\draw[] (9*\mydistx, 0 + 0*\mydisty) rectangle (10*\mydistx, 1 + 0*\mydisty) {};
\node[](d) at (9.5*\mydistx, 0.5 + 0*\mydisty)  {\footnotesize $\bullet\bullet\bullet$};     

\draw[pattern=vertical lines] (10*\mydistx, 0 + 0*\mydisty) rectangle (11*\mydistx, 1 + 0*\mydisty) {};
\node[](d) at (10.5*\mydistx, -0.3 + 0*\mydisty)  {$\db_{v_1,v_2}$};

\node[](d) at (-0.5*\mydistx, 0.5 + 1*\mydisty)  {$B_{v_1}$};
	
\draw[red, line width=3pt] (0.5*\mydistx, 0 + 1*\mydisty) rectangle (0.5*\mydistx, 1 + 1*\mydisty) {};
\draw[] (0*\mydistx, 0 + 1*\mydisty) rectangle (1*\mydistx, 1 + 1*\mydisty) {};
\draw[pattern=vertical lines] (0*\mydistx, 0 + 1*\mydisty) rectangle (0.5*\mydistx, 1 + 1*\mydisty) {};
\draw[pattern=north west lines, pattern color=skyblue] (0*\mydistx, 1 + 1*\mydisty) rectangle (1*\mydistx, 2 + 1*\mydisty) {};
\node[](d) at (0.5*\mydistx, -0.3 + 1*\mydisty)  {$\di_{v_1,2}$};
\draw[red, line width=3pt] (0*\mydistx, 1 + 1*\mydisty) rectangle (0.5*\mydistx, 1 + 1*\mydisty) {};

\draw[] (1*\mydistx, 0 + 1*\mydisty) rectangle (2*\mydistx, 1 + 1*\mydisty) {};
\draw[pattern=vertical lines] (1*\mydistx, 0 + 1*\mydisty) rectangle (1.5*\mydistx, 1 + 1*\mydisty) {};
\draw[red, line width=3pt] (1.5*\mydistx, 0 + 1*\mydisty) rectangle (1.5*\mydistx, 1 + 1*\mydisty) {};
\draw[pattern=north west lines, pattern color=skyblue] (1*\mydistx, 1 + 1*\mydisty) rectangle (2*\mydistx, 2 + 1*\mydisty) {};
\node[](d) at (1.5*\mydistx, -0.3 + 1*\mydisty)  {$\di_{v_1,3}$};
\draw[red, line width=3pt] (1*\mydistx, 1 + 1*\mydisty) rectangle (1.5*\mydistx, 1 + 1*\mydisty) {};

\draw[pattern=vertical lines] (2*\mydistx, 0 + 1*\mydisty) rectangle (3*\mydistx, 1 + 1*\mydisty) {};
\node[](d) at (2.5*\mydistx, -0.3 + 1*\mydisty)  {$\di_{v_2,1}$};
\draw[pattern=vertical lines] (3*\mydistx, 0 + 1*\mydisty) rectangle (4*\mydistx, 1 + 1*\mydisty) {};
\node[](d) at (3.5*\mydistx, -0.3 + 1*\mydisty)  {$\di_{v_2,3}$};

\draw[] (4*\mydistx, 0 + 1*\mydisty) rectangle (5*\mydistx, 1 + 1*\mydisty) {};
\node[](d) at (4.5*\mydistx, 0.5 + 1*\mydisty)  {\footnotesize $\bullet\bullet\bullet$};

\draw[pattern=north west lines, pattern color=skyblue] (5*\mydistx, 0 + 1*\mydisty) rectangle (6*\mydistx, 1 + 1*\mydisty) {};
\node[](d) at (5.5*\mydistx, -0.3 + 1*\mydisty)  {$\dy_{v_1}$};

\draw[pattern=vertical lines] (6*\mydistx, 0 + 1*\mydisty) rectangle (7*\mydistx, 1 + 1*\mydisty) {};
\node[](d) at (6.5*\mydistx, -0.3 + 1*\mydisty)  {$\dy_{v_2}$};

\draw[] (7*\mydistx, 0 + 1*\mydisty) rectangle (8*\mydistx, 1 + 1*\mydisty) {};
\node[](d) at (7.5*\mydistx, 0.5 + 1*\mydisty)  {\footnotesize $\bullet\bullet\bullet$};

\draw[pattern=vertical lines] (8*\mydistx, 0 + 1*\mydisty) rectangle (9*\mydistx, 1 + 1*\mydisty) {};
\node[](d) at (8.5*\mydistx, -0.3 + 1*\mydisty)  {$\dy_{v_1,v_2}$};

\draw[] (9*\mydistx, 0 + 1*\mydisty) rectangle (10*\mydistx, 1 + 1*\mydisty) {};
\node[](d) at (9.5*\mydistx, 0.5 + 1*\mydisty)  {\footnotesize $\bullet\bullet\bullet$};  

\draw[pattern=vertical lines] (10*\mydistx, 0 + 1*\mydisty) rectangle (11*\mydistx, 1 + 1*\mydisty) {};
\node[](d) at (10.5*\mydistx, -0.3 + 1*\mydisty)  {$\db_{v_1,v_2}$};

\node[](d) at (-0.5*\mydistx, 0.5 + 2*\mydisty)  {$B_{v_2}$};
	
\draw[pattern=vertical lines] (0*\mydistx, 0 + 2*\mydisty) rectangle (1*\mydistx, 1 + 2*\mydisty) {};
\node[](d) at (0.5*\mydistx, -0.3 + 2*\mydisty)  {$\di_{v_1,2}$};
\draw[pattern=vertical lines] (1*\mydistx, 0 + 2*\mydisty) rectangle (2*\mydistx, 1 + 2*\mydisty) {};
\node[](d) at (1.5*\mydistx, -0.3 + 2*\mydisty)  {$\di_{v_1,3}$};

\draw[] (2*\mydistx, 0 + 2*\mydisty) rectangle (3*\mydistx, 1 + 2*\mydisty) {};
\draw[pattern=vertical lines] (2*\mydistx, 0 + 2*\mydisty) rectangle (2.5*\mydistx, 1 + 2*\mydisty) {};
\draw[red, line width=3pt] (2.5*\mydistx, 0 + 2*\mydisty) rectangle (2.5*\mydistx, 1 + 2*\mydisty) {};
\node[](d) at (2.5*\mydistx, -0.3 + 2*\mydisty)  {$\di_{v_2,1}$};
\draw[pattern=north east lines, pattern color=blue] (2*\mydistx, 1 + 2*\mydisty) rectangle (3*\mydistx, 2 + 2*\mydisty) {};
\draw[red, line width=3pt] (2*\mydistx, 1 + 2*\mydisty) rectangle (2.5*\mydistx, 1 + 2*\mydisty) {};

\draw[] (3*\mydistx, 0 + 2*\mydisty) rectangle (4*\mydistx, 1 + 2*\mydisty) {};
\draw[pattern=vertical lines] (3*\mydistx, 0 + 2*\mydisty) rectangle (3.5*\mydistx, 1 + 2*\mydisty) {};
\draw[red, line width=3pt] (3.5*\mydistx, 0 + 2*\mydisty) rectangle (3.5*\mydistx, 1 + 2*\mydisty) {};
\node[](d) at (3.5*\mydistx, -0.3 + 2*\mydisty)  {$\di_{v_2,3}$};
\draw[pattern=north east lines, pattern color=blue] (3*\mydistx, 1 + 2*\mydisty) rectangle (4*\mydistx, 2 + 2*\mydisty) {};
\draw[red, line width=3pt] (3*\mydistx, 1 + 2*\mydisty) rectangle (3.5*\mydistx, 1 + 2*\mydisty) {};

\draw[] (4*\mydistx, 0 + 2*\mydisty) rectangle (5*\mydistx, 1 + 2*\mydisty) {};
\node[](d) at (4.5*\mydistx, 0.5 + 2*\mydisty)  {\footnotesize $\bullet\bullet\bullet$};

\draw[pattern=vertical lines] (5*\mydistx, 0 + 2*\mydisty) rectangle (6*\mydistx, 1 + 2*\mydisty) {};
\node[](d) at (5.5*\mydistx, -0.3 + 2*\mydisty)  {$\dy_{v_1}$};

\draw[pattern=north east lines, pattern color=blue] (6*\mydistx, 0 + 2*\mydisty) rectangle (7*\mydistx, 1 + 2*\mydisty) {};
\node[](d) at (6.5*\mydistx, -0.3 + 2*\mydisty)  {$\dy_{v_2}$};

\draw[] (7*\mydistx, 0 + 2*\mydisty) rectangle (8*\mydistx, 1 + 2*\mydisty) {};
\node[](d) at (7.5*\mydistx, 0.5 + 2*\mydisty)  {\footnotesize $\bullet\bullet\bullet$};

\draw[pattern=vertical lines] (8*\mydistx, 0 + 2*\mydisty) rectangle (9*\mydistx, 1 + 2*\mydisty) {};
\node[](d) at (8.5*\mydistx, -0.3 + 2*\mydisty)  {$\dy_{v_1,v_2}$};

\draw[] (9*\mydistx, 0 + 2*\mydisty) rectangle (10*\mydistx, 1 + 2*\mydisty) {};
\node[](d) at (9.5*\mydistx, 0.5 + 2*\mydisty)  {\footnotesize $\bullet\bullet\bullet$}; 

\draw[pattern=vertical lines] (10*\mydistx, 0 + 2*\mydisty) rectangle (11*\mydistx, 1 + 2*\mydisty) {};
\node[](d) at (10.5*\mydistx, -0.3 + 2*\mydisty)  {$\db_{v_1,v_2}$};

\node[](d) at (-0.5*\mydistx, 0.5 + -1*\mydisty)  {$B_{v_1,v_2}$};
	
\draw[] (0*\mydistx, 0 + -1*\mydisty) rectangle (1*\mydistx, 1 + -1*\mydisty) {};
\draw[pattern=north east lines, pattern color=vermilion] (0*\mydistx, 0 + -1*\mydisty) rectangle (0.5*\mydistx, 1 + -1*\mydisty) {};
\node[](d) at (0.5*\mydistx, -0.3 + -1*\mydisty)  {$\di_{v_1,2}$};

\draw[pattern=vertical lines] (1*\mydistx, 0 + -1*\mydisty) rectangle (2*\mydistx, 1 + -1*\mydisty) {};
\node[](d) at (1.5*\mydistx, -0.3 + -1*\mydisty)  {$\di_{v_1,3}$};

\draw[] (2*\mydistx, 0 + -1*\mydisty) rectangle (3*\mydistx, 1 + -1*\mydisty) {};
\node[](d) at (2.5*\mydistx, -0.3 + -1*\mydisty)  {$\di_{v_2,1}$};
\draw[pattern=north east lines, pattern color=vermilion] (2*\mydistx, 0 + -1*\mydisty) rectangle (2.5*\mydistx, 1 + -1*\mydisty) {};

\draw[pattern=vertical lines] (3*\mydistx, 0 + -1*\mydisty) rectangle (4*\mydistx, 1 + -1*\mydisty) {};
\node[](d) at (3.5*\mydistx, -0.3 + -1*\mydisty)  {$\di_{v_2,3}$};

\draw[] (4*\mydistx, 0 + -1*\mydisty) rectangle (5*\mydistx, 1 + -1*\mydisty) {};
\node[](d) at (4.5*\mydistx, 0.5 + -1*\mydisty)  {\footnotesize $\bullet\bullet\bullet$};

\draw[pattern=vertical lines] (5*\mydistx, 0 + -1*\mydisty) rectangle (6*\mydistx, 1 + -1*\mydisty) {};
\node[](d) at (5.5*\mydistx, -0.3 + -1*\mydisty)  {$\dy_{v_1}$};

\draw[pattern=vertical lines] (6*\mydistx, 0 + -1*\mydisty) rectangle (7*\mydistx, 1 + -1*\mydisty) {};
\node[](d) at (6.5*\mydistx, -0.3 + -1*\mydisty)  {$\dy_{v_2}$};

\draw[] (7*\mydistx, 0 + -1*\mydisty) rectangle (8*\mydistx, 1 + -1*\mydisty) {};
\node[](d) at (7.5*\mydistx, 0.5 + -1*\mydisty)  {\footnotesize $\bullet\bullet\bullet$};

\draw[] (8*\mydistx, 0 + -1*\mydisty) rectangle (9*\mydistx, 1 + -1*\mydisty) {};
\node[](d) at (8.5*\mydistx, -0.3 + -1*\mydisty)  {$\dy_{v_1,v_2}$};

\draw[] (9*\mydistx, 0 + -1*\mydisty) rectangle (10*\mydistx, 1 + -1*\mydisty) {};
\node[](d) at (9.5*\mydistx, 0.5 + -1*\mydisty)  {\footnotesize $\bullet\bullet\bullet$}; 

\draw[pattern=north east lines, pattern color=vermilion] (10*\mydistx, 0 + -1*\mydisty) rectangle (11*\mydistx, 1 + -1*\mydisty) {};
\node[](d) at (10.5*\mydistx, -0.3 + -1*\mydisty)  {$\db_{v_1,v_2}$};
     	
\end{tikzpicture}
}
\end{minipage}
\caption{Schematic illustration of the reduction of \Cref{thm-w-hardness-ls-bin-packing}. 
$a)$ shows the initial assignment~$f$ and~$b)$ shows the improving assignment~$f'$ with~$d_{\rm flip}(f,f') \le k$.
The squares in each bin correspond to a set of dimensions, for example~$\di_{v_1,2}$ corresponds to all important dimensions with respect to~$(v_1,2)$.
The fillings of the squares correspond to vectors which are contained in that bin according to the assignment~$f$ or~$f'$.
All \textcolor{skyblue}{skyblue} filed squares correspond to the dimensions in which the vertex vector~$\vect(v_1)$ has value~1.
Analogously, \textcolor{blue}{blue} corresponds to vertex vector~$\vect(v_2)$, \textcolor{orange}{orange} corresponds to edge vector~$\vect(v_1,v_2)$, and \textcolor{vermilion}{vermilion} corresponds to the blocking vector~$\bvect(v_1,v_2)$.
Furthermore, all squares filled with vertical black lines indicate that the corresponding bin contains dummy vectors in each of the dimensions.
The red vertical line, for example in square~$\di_{v_1,2}$ in~$B_{v_1}$ indicates that only for dimensions~$\di_{v_1,2}^y$ with~$y\in[1,z/2-1]$ there is a dummy vector.
In~$a)$ for assignment~$f$ observe that only edge bins have an overload.
More precisely, $B_{v_1,v_2}$ has an overload of~1 in dimensions~$\di_{v_1,2}^y$ and~$\di_{v_2,1}^y$ for~$y\in[1,z/2]$.
In contrast, in~$b)$, the assignment~$f'$, (1) the vertex bins containing its corresponding vertex vector, and (2) the edge bins containing its corresponding edge vector, have overload.
For example bin~$B_{v_1}$ has overload in dimensions~$\di_{v_1,2}^y$ for~$y\in[1,z/2-1]$.
Note that the distinction of~$z/2$ in~$f$ to~$z/2-1$ in~$f'$ is essential to show that~$f'$ is improving upon~$f$.}
\label{fig:example-w-hardness-bin-packing}
\end{figure}

\textbf{Construction:}
First, we set the search radius $k\coloneqq \ell + \binom{\ell}{2}$.
Second, we describe the dimensions, third we describe the bins and the corresponding weight vector~$\omega$, fourth we describe all vectors  and their initial partition to the bins, and finally we show the desired bounds on the parameters. 
An illustration is shown in \Cref{fig:example-w-hardness-bin-packing}.

\emph{Description of dimensions:}
Initially, we describe the \emph{important} dimensions.
Let~$z\coloneqq 4\cdot\ell^2$.
For each color~$i$ and each vertex~$v_i\in V_i$, each other color~$j\in[1,\ell]\setminus\{i\}$, and each~$y\in[1,z]$ we add an \emph{important dimension}~$\di_{v_i,j}^y$.
All dimensions~$\di_{v_i,j}^1,\ldots, \di_{v_i,j}^z$ are referred to as the \emph{important dimensions with respect to~$(v_i,j)$}.
Additionally, dimensions~$\di_{v_i,j}^1, \ldots , \di_{v_i,j}^{z/2}$ are also called the \emph{important profit dimensions of~$(v_i,j)$}.

Next, we describe the \emph{dummy dimensions}.
Let~$Z\coloneqq 2\cdot k\cdot z$.
For each color~$i$ and each vertex~$v_i\in V_i$ and each~$y\in[1,Z]$ we add a \emph{dummy vertex dimension}~$\dy_{v_i}^y$.
All dimensions~$\dy_{v_i}^1,\ldots, \dy_{v_i}^{Z}$ are referred to as the \emph{dummy vertex dimensions with respect to~$v_i$}.
Furthermore, for each edge~$\{u_i,w_j\}\in E(G)$ such that~$u_i\in V_i$ and~$w_j\in V_j$ (note that~$i\ne j$) and each~$y\in[1,Z]$ we add a \emph{dummy edge dimension}~$\dy_{u_i,w_j}^y$ and a \emph{dummy blocking dimension}~$\db_{u_i,w_j}^y$.
All dimensions~$\dy_{u_i,w_j}^1,\ldots, \dy_{u_i,w_j}^{Z}$ are referred to as the \emph{dummy edge dimensions with respect to edge~$\{u_i,w_j\}$} and all dimensions~$\db_{u_i,w_j}^1,\ldots, \db_{u_i,w_j}^{Z}$ are referred to as the \emph{dummy blocking dimensions with respect to edge~$\{u_i,w_j\}$}.

\emph{Description of bins and weight-vector~$\omega$:}
Our \BinLoc instance contains a \emph{target bin}~$B^*$, for each color~$i$ and each vertex~$v_i\in V_i$ a \emph{vertex bin}~$B_{v_i}$, and for each edge~$\{u_i,w_j\}$ with~$u_i\in V_i$ and~$w_j\in V_j$ (note that~$i\ne j$) of~$G$ an \emph{edge bin}~$B_{u_i,w_j}$.
All bins have the target weight vector~$\omega$ which has value~1 in each dimensions.

\emph{Vectors and their initial assignment~$f$ to the bins:}
For an intuition of~$f$, we refer to \Cref{fig:example-w-hardness-bin-packing}.
For each vertex~$v$ in any color class, say~$V_i$, we add a \emph{vertex vector}~$\vect(v_i)$ having value~1 in each important dimension with respect to~$(v,j)$ where~$j\in[1,\ell]\setminus\{i\}$ and in each dummy vertex dimension with respect to~$v$.
Furthermore, for each edge~$\{u_i,w_j\}\in E(G)$ (note that~$i\ne j$) we add an \emph{edge vector}~$\vect(u_i,w_j)$ having value~$1$ in each important dimension of~$(u_i,j)$, each important dimension~$(w_j,i)$, and each dummy edge dimension with respect to~$\{u_i,w_j\}$.
Additionally, for each edge~$\{u_i,w_j\}\in E(G)$ we add a \emph{blocking vector}~$\bvect(u_i,w_j)$ having value~$1$ in each important profit dimension of~$(u_i,j)$, each important profit dimension~$(w_j,i)$, and each dummy blocking dimension with respect to edge~$\{u_i, w_j\}$.
Analogously to the edge vectors, observe that the second index in the important profit dimensions where the blocking vector has value~1 is always the color class of the other vertex of that edge.

Next, we describe the mapping~$f$ which assign each vector to one bin.
To avoid confusion with the colors of the graph~$G$, here we abuse our notation and let~$f$ map to bins instead of colors as required by our definition.
However, one can use any bijection from the bins to a set of $n+m+1$~colors to fulfill the definition.

\begin{itemize}
\item For each vertex vector~$\vect(v_i)$, we set~$f(\vect(v_i))=B^*$

\item For each edge vector~$\vect(u_i,w_j)$, we set~$f(\vect(u_i,w_j))=B_{u_i,w_j}$

\item For each blocking vector~$\bvect(u_i,w_j)$, we set~$f(\bvect(u_i,w_j))B_{u_i,w_j}$
\end{itemize}

In the following we call a dimension~$\di$ with respect to a bin~$B$ \emph{empty} if each vector~$\vect$ contained in~$B$ has value~0 in dimension~$\di$.
Finally, we add many \emph{dummy vectors} to all bins to ensure that only a few number of dimensions are empty.
Here, a dummy vector has value~1 in exactly one dimension and value~0 in each other dimension.
For a dimensions~$\di$ by~$\vect(\di)$ we denote a dummy vector having value~1 exactly in dimension~$\di$.
Next, we add several dummy vector to the instance.
Note that many dummy vectors are identical, but each two identical dummy vectors are assigned to different bins by~$f$.

\begin{itemize}
\item For the target bin~$B^*$ we do the following:
\begin{itemize}

\item For each dummy blocking dimension~$\db$, we add a vector~$\vect(\db)$, and we set~$f(\vect(\db))=B^*$.
\end{itemize}

\item For each vertex bin~$B_{v_i}$ corresponding to vertex~$v_i$ in color class~$i$ we do the following:
\begin{itemize}
\item For each~$j\in[1,\ell]\setminus\{i\}$, and each~$y\in[1,z/2-1]$, that is, in all important profit dimensions of~$(v_i,j)$ \emph{except} the last one~$d_{u_i}^{z/2}$, we add a dummy vector~$\vect(d_{v_i}^y))$, and we set~$f(d_{v_i}^y))=B_{v_i}$.

This offset of~1 in the important profit dimensions is essential to reduce the overload.
More precisely, only for these dimensions there exists a bin (this vertex bin~$B_{v_i}$) which is empty and also at least one other bin which contains at least two vectors which have value~1 in that dimension.

\item For each color~$j$, and each vertex~$u_j\in V_j$ such that~$v_i\ne u_j$, and each~$p\in[1,\ell]\setminus\{j\}$, we add a dummy vector~$\vect(\di)$ for each important dimension with respect to~$(u_j,p)$, and each dummy dimensions \emph{except} the dummy vertex dimensions with respect to~$v_i$, and we set~$f(\vect(\di))=B_{v_i}$.
\end{itemize}

\item For each edge bin~$B_{u_i,w_j}$ corresponding to edge~$(u_i,w_j)$ where~$u_i\in V_i$ and~$w_j\in V_j$ (note that~$i\ne j)$ we do the following:
\begin{itemize}

\item We add a dummy vector~$\vect(\di)$ for each important dimension  with respect to~$(v,p)$ where~$v_p\ne u_j$ and~$v_p\ne w_i$.
We set~$f(\vect(\di))=B_{u_i.w_j}$.

\item We add a dummy vector~$\vect(\di)$ for each dummy dimensions~$\di$ \emph{except} the dummy edge dimension with respect to~$\{u_i,w_j\}$ and the blocking edge dimension with respect to~$\{u_i,w_j\}$.
We set~$f(\vect(\di))=B_{u_i,w_j}$.
\end{itemize}

\end{itemize}

In the following, a dimension~$\di$ is called \emph{overfull} if~$\di$ contributes at least~1 to the overload, that is, there exists a bin~$B$ and at least two vectors~$\vect_1$ and~$\vect_2$ contained in~$B$ such that~$\vect_1[\di]=1=\vect_1[\di]$.
To reduce the overload, we need to reassign one vector, say~$\vect_1$, having value~1 in an overfull dimension~$\di$ to another bin in which~$\di$ is empty.
Now, we make an observation which dimensions are overfull for the assignment~$f$ and a subsequent observation in which bins these dimensions are empty (also see \Cref{fig:example-w-hardness-bin-packing}).
We want to emphasize that there are way more empty (dummy) dimensions with respect to some bins, but these dimensions are not necessary to reduce the overload since they are not overfull in any bin.

\begin{observation}
\label{obs-bin-packing-initial-assignment}
\begin{enumerate}
For the initial assignment~$f$ of the vectors to the bins the following holds:
\item 
\label{obs-bin-packing-initial-assignment-overload}
Only the edge bins have an overload of at least~1.
More precisely, for each edge~$\{u_i,w_j\}\in E(G)$ where~$u_i\in V_i$ and~$w_j\in V_j$ in edge bin~$B_{u_i,w_j}$ each important profit dimension~$d_{u_i,j}^y$ and~$d_{w_j,i}^y$ for each~$y\in[1,z/2]$ has an overload of~1.

This overload is due to the edge vector~$\vect(u_i,w_j)$ and the blocking vector~$\bvect(u_i,w_j)$.

\item
\label{obs-bin-packing-initial-assignment-empty}
Only the vector bins have empty important profit dimensions.
More precisely, for each color~$i$ and each~$v_i\in V_i$ in vector bin~$B_{v_i}$ the important dimensions~$d_{v_i,j}^y$ for each~$j\in[1,\ell]\setminus\{i\}$ and each~$y\in[z/2,z]$ are empty.

\end{enumerate}
\end{observation}

Let us remark that dimensions~$d_{v,j}^{z/2}$ for each~$v_i\in V(G)$ and each~$j\in[1,\ell]\setminus\{i\}$ are the only dimensions which are overfull in one bin (more precisely: in an edge bin corresponding to an edge incident with~$v_i$) and empty in another bin (more precisely: in the vertex bin~$B_{v_i}$).
By this observation, one can obtain a bound on~$2k$ for the maximal decrease of the overload by any $k$-flip.
This is due to the fact that vector has value one in at most two dimensions~$d_{v,j}^{z/2}$ (more precisely: only edge vectors and blocking edge vectors have exactly two dimensions~$d_{v,j}^{z/2}$ with value~1).
For showing the correctness we do not use this argument; instead we use the weaker observation in how many dimensions a move can increase the overload, to allow for easier arguments about the moves of the vectors in an improving $k$-flip.

\emph{Parameter bounds.}
First, we verify that the number of ones per vector is small.
\begin{itemize}
\item Each vertex vector has value~1 in $(\ell-1)\cdot z+Z\in\OO(\ell^4)$~dimensions.

\item Each edge vector has value~1 in $2\cdot z+Z\in\OO(\ell^4)$~dimensions.

\item Each blocking vector has value~$1$ in~$z/2+z/2+Z\in\OO(\ell^4)$~dimensions.

\item Each dummy vector has value~1 in exactly one dimension.
\end{itemize}

Thus, the maximal number~$q$ of ones per vector is bounded solely on~$\ell$.
Since each vector has only entries~0 and~1, the maximal sum of entries over all vectors is also bounded in~$\ell$.

Note that we have $n+m+1$~bins.
Furthermore, we have $z\cdot n\cdot (\ell-1)$~important dimensions and $(n+m+m)\cdot Z$~dummy dimension.
Since we have one vertex vector per vertex, one edge vector and one blocking vector per edge, and at most one dummy vector per dimension and per bin, the instance of \BinLoc can be computed in time polynomial in the input size.

\textbf{Correctness:}
We show that~$G$ has a multicolored clique if and only if there is
another assignment~$f'$ of the vectors to the bins such that~$d_{\rm flip}(f,f') \le k$ which is improving upon~$f$.

$(\Rightarrow)$
Let~$C=(v_1,\ldots, v_\ell)$ be a multicolored clique of~$G$.
We set~$$D_{\rm flip}(f,f')\coloneqq (\bigcup_{i\in\ell} \vect(v_i))\cup \bigcup_{i,j\in[1,\ell], i\ne j} \vect(v_i,v_j).$$ 

More precisely, compared to~$f$ in the new assignment~$f'$, we set~$f(\vect(v_i))=B_{v_i}$ for each~$v_i\in C$ and we set~$f'(\vect(v_i,v_j))=B^*$ for each edge~$\{v_i,v_j\}$ having both endpoints in~$C$.
Hence, it remains to verify that the new assignment of vectors to bins has a smaller overload than the initial assignment.
Recall that~$f(\vect(v_i))=B^*$ for each~$v_i\in C$ and~$f(\vect(v_i,v_j))=B_{v_i,v_j}$ for each edge having both endpoints in~$C$.

To show that~$f'$ is improving against~$f$, we first analyze the change in the overload by only considering swapping all vertex vectors corresponding to vertices in~$C$.
Second, we analyze which dimensions are empty in the target bin after this partial swap.
Third, we analyze the influence on the overload of the subsequent swap of the edge vectors corresponding to edges of~$C$ into the target bin.
We then conclude that~$f'$ is improving over~$f$.

\emph{Step 1: Change of the overload by the movement of all vertex vectors corresponding to vertices of~$C$.}
First, we analyze the important dimensions:
According to \Cref{obs-bin-packing-initial-assignment-empty} of \Cref{obs-bin-packing-initial-assignment}, in each vertex bin~$B_{v_i}$ only the important dimensions~$d_{v_i,j}^y$ for~$j\in[1,\ell]\setminus\{i\}$ and~$y\in[z/2,z]$ are empty, and also according to \Cref{obs-bin-packing-initial-assignment-overload} of \Cref{obs-bin-packing-initial-assignment} initially there is no overload in any dimension of the target bin~$B^*$.
Second, we analyze the dummy dimensions:
Vertex vector~$\vect(v_i)$ has value~1 in each dummy vertex dimension with respect to~$v_i$ and in all these dimension the target bin~$B^*$ has no overload and the vertex bin~$B_{v_i}$ is empty in all these dimensions.
Thus, moving the vertex vector~$\vect(v_i)$ from the target bin~$B^*$ into its corresponding vertex bin~$B_{v_i}$ increases the overload by exactly~$(\ell-1)\cdot (z-(z/2+1)=(\ell-1)\cdot(z/2-1)$.
Consequently, after each vertex vector corresponding to a vertex of~$C$ has been moved from the target bin~$B^*$ into its corresponding vertex bin the overload increased by exactly~$\ell\cdot (\ell-1)\cdot (z/2-1)$.

\emph{Step 2: Observation for target bin~$B^*$, after all these vertex vectors have been moved out of~$B^*$.}
Observe that each dimensions in which any of the moved vertex vectors has value~1 is now empty.
In particular, for each~$v_i\in C$ and each~$j\in[1,\ell]\setminus\{i\}$ all important dimensions with respect to~$(v_i,j)$ are now empty in~$B^*$.
Moreover, also the corresponding dummy dimensions of~$v_i$ in the target bin~$B^*$ are now empty.

\emph{Step 3: Change of the overload by the movement of all edge vectors corresponding to edges having both endpoints in~$C$.}
According to \Cref{obs-bin-packing-initial-assignment-overload} of \Cref{obs-bin-packing-initial-assignment} initially in each edge bin~$B_{u_i,w_j}$ each important profit dimension with respect to~$(u_i,j)$ and~$(w_j,i)$ has on overload of~$1$.
Thus, moving the edge vector~$\vect(u_i,w_j)$ out of the edge bin~$B_{u_i,w_j}$ reduces the overload by~$2\cdot z/2=z$.
In total, this reduces the overload by~$\binom{\ell}{2}\cdot z$.

Recall that edge vector~$\vect(u_i,w_j)$ has value~1 in each important dimension with respect to~$(u_i,j)$ and with respect to~$(w_j,i)$, and in each dummy edge dimension with respect to~$\{u_i, w_j\}$.
Hence, any two distinct edges having both endpoints in~$C$ have value~$1$ in different dimensions.
Now, we analyze the change of the overload of moving all these edge vectors into the target bin~$B^*$.
First, we investigate the important dimensions:
After all vertex vectors corresponding to vertices in~$C$ have been moved out of the target bin~$B^*$, all important dimensions with respect to~$(v_i,j)$ where~$v_i\in C$ and~$j\in[1,\ell]\setminus\{i\}$ are empty.
Thus, moving all edge vectors corresponding to edges having both endpoints in~$C$ into~$B^*$ does not increase the overload.
Second, we analyze the dummy dimensions.
Since all edge dummy dimensions are empty with respect to~$B^*$ and since each two distinct edge vectors have value~1 in distinct dummy edge dimensions, the overload does not increase.

Hence, in total we reduce the overload by~$\binom{\ell}{2}\cdot z- \ell(\ell-1)\cdot (z/2-1)=\ell(\ell-1)>0$ and thus~$f'$ is improving over~$f$.

$(\Leftarrow)$
Let~$f'$ be another assignment of the vectors to the bins such that~$d_{\rm flip}(f,f') \le k$ which is improving upon~$f$.
Furthermore, we assume that~$d_{\rm flip}(f,f')$ is smallest among all improving solutions upon~$f$.

\emph{Outline:}
\begin{enumerate}
\item
First, we provide a tool to show that~$f'$ cannot move certain vectors~$\vect$ in certain bins~$B$, that is, $f'(\vect)\ne B$, since otherwise we can construct another improving assignment~$f^*$ upon~$f$ which strictly moves less vectors than~$f$, contradicting our assumption on~$f'$.

\item
According to \Cref{obs-bin-packing-initial-assignment-overload} of \Cref{obs-bin-packing-initial-assignment}, in order to reduce the overload, some vectors assigned by~$f$ to edge bins have to be contained in~$D_{\rm flip}(f,f')$.
More precisely, for an edge bin~$B_{u_i,w_j}$ we can only reduce the overload if edge vector~$\vect(u_i,w_j)\in D_{\rm flip}(f,f')$ or if blocking vector~$\bvect(u_i,w_j)\in D_{\rm flip}(f,f')$.
Second, we use this tool to show the latter is not possible, that is, $D_{\rm flip}(f,f')$ contains no blocking vector.

\item
The second step then implies that some edge vectors have to be contained in~$D_{\rm flip}(f,f')$.
Third, we use the tool to show that each edge vector~$\vect(u_i,w_j)\in D_{\rm flip}(f,f')$ is moved into the target bin~$B^*$.

\item
Observe that in the initial assignment~$f$ no important dimensions is empty in the target bin~$B^*$.
More precisely, in the assignment~$f$ each important dimension in~$B^*$ has value~1 since exactly~1 vertex vector has value~1 in that dimension.
Hence, if we only consider the move of edge vectors in~$D_{\rm flip}(f,f')$, then the overload increases by each such move.
Since~$f'$ is improving upon~$f$, $f'$ assigns some vertex vectors from the target bin~$B^*$ to different bins.
Fourth, again by using the tool from Step~1, we verify that each vertex vector~$\vect(v_i)\in D_{\rm flip}(f,f')$ is moved into its corresponding vertex bin~$B_{v_i}$.

\item 
Fifth, we show that no dummy vector is contained in~$D_{\rm flip}(f,f')$.

\item 
Steps 2 to 5 provide many properties of~$D_{\rm flip}(f,f')$ and~$f'$:
By Steps 2 and 5, we know that each vector~$\vect\in D_{\rm flip}(f,f')$ is an edge vector or a vertex vector.
Furthermore, by the third step, we know that for each edge vector~$\vect(u_i,w_j)\in D_{\rm flip}(f,f')$ we have~$f'(\vect(u_i,w_j))=B^*$, and by the fourth step we know that each vertex vector~$\vect(v_i)\in D_{\rm flip}(f,f')$ we have~$f'(\vect(v_i))=B_{v_i}$.
Sixth, these properties together with the assumption that~$d_{\rm flip}(f,f')$ is smallest among all improving flips, allow us to show that  if an edge vector~$\vect(u_i,w_j)\in D_{\rm flip}(f,f')$ then also both vertex vectors~$\vect(u_i)\in D_{\rm flip}(f,f')$ and~$\vect(w_j)\in D_{\rm flip}(f,f')$.

\item
Finally, in the seventh step we put everything together and verify that all vertex vectors contained in~$D_{\rm flip}(f,f')$ correspond to the vertices of a multicolored clique in~$G$ and that all edge vectors contained in~$D_{\rm flip}(f,f')$ correspond to the edges of that multicolored clique.
\end{enumerate}

\emph{Step 1: Tool that some vector~$\vect$ cannot be moved in some bin~$B$ by~$f'$.}
Roughly speaking, if~$\vect$ has too many dimensions with value~$1$ which are filled in~$B$ with dummy vectors, then we construct another improving solution~$f''$ by setting~$f''(\vect)\coloneqq f(\vect)$ and~$f''$ is identical to~$f'$ for all other vectors.
This then contradicts the assumption that~$d_{\rm flip}(f,f')$ is smallest among all improving flips.

Formally, for a vector~$\vect$, with~$f(\vect)=B_1$ for some bin~$B_1$, we call a subset~$\mathcal{D}_{f(\vect)}^{B_2}$ of the dimensions in which~$\vect$ has value~1 a \emph{conflicting dimension set of~$\vect$ with respect to (another bin)~$B_2$} if for each dimensions~$\di\in\mathcal{D}_{f(\vect)}^{B_2}$ in the initial assignment~$f$, bin~$B_2$ contains a dummy vector which has value~1 in dimension~$\di$.

\begin{claim}
\label{claim-bound-forbidden-move}
If a vector~$\vect$ has a conflicting dimension set with respect to some bin~$B$ of size at least~$Z$, then~$f'(\vect)\ne B$.
\end{claim}
\begin{claimproof}
We show that if~$f'(\vect)=B$, then~$f'$ cannot be improving upon~$f$.
To verify this, we first give a lower bound on the increase of the overload by assigning vector~$\vect$ to bin~$B$ and second, we provide an upper bound on the maximal overload reduction of any improving $k$-flip.
We then show that the lower bound is larger than the upper bound, yielding a contradiction.

\emph{Step A:}
By the definition of the conflicting dimension set~$\mathcal{D}_{f(\vect)}^{B}$, setting~$f'(\vect)=B$ increases the overload by at least~$Z$.
Since for each dimension~$\di\in\mathcal{D}_{f(\vect)}^{B}$, bin~$B$ contains a dummy vector having value~1 in dimension~$\di$, the overload of the $k$-flip~$f$ with respect to dimensions~$\mathcal{D}_{f(\vect)}^{B}$ in bin~$B$ increases by at least~$Z-2k$.

\emph{Step B:}
According to \Cref{obs-bin-packing-initial-assignment-overload} of \Cref{obs-bin-packing-initial-assignment}, the only bins in the initial assignments having overfull dimensions are the edge bins.
More precisely, each edge bin initially has an overload of exactly~$z$.
Hence, to reduce the overload of the initial assignment, we have to move vectors out of edge bins.
Since each such move can decrease the overload by at most~$z$ and since we can move at most $k$~vectors, in total the overload can be decrease by at most~$k\cdot z$.

Now, since~$|\mathcal{D}_{f(\vect)}^{B}|\ge Z=2\cdot k\cdot z>k\cdot z+2k$, the overload of the new solution~$f'$ is larger than the overload of the initial solution~$f$, a contradiction.
Thus, $f'(\vect)\ne B$.
\end{claimproof}

\emph{Step 2: No blocking vector is contained in~$D_{\rm flip}(f,f')$}.
We use \Cref{claim-bound-forbidden-move} to show this property.

\begin{claim}
\label{claim-do-not-move-blocking}
For each blocking vector~$\bvect(u_i,w_j)$, we have~$\bvect(u_i,w_j)\notin D_{\rm flip}(f,f')$.
\end{claim}
\begin{claimproof}
We assume towards a contradiction that for at least one blocking vector we have $\bvect(u_i,w_j)\in D_{\rm flip}(f,f')$.
Recall that~$f(\bvect(u_u,w_j))=B_{u_i,w_j}$.
Observe that~$\bvect(u_i,w_j)$ has value~1 in each dummy blocking dimension with respect to edge~$\{u_i,w_j\}$.
Also, note that in the initial assignment~$f$, no bin is empty in any dummy blocking dimension.
Furthermore, by construction, each vector except vector~$\bvect(u_i,w_j)$ having value~1 in at least one dummy blocking dimension with respect to edge~$\{u_i,w_j\}$ in any bin is a dummy vector.
Thus, each such vector has exactly one~$1$ in one dummy blocking dimension with respect to edge~$\{u_i,w_j\}$.
Consequently, for each bin~$B$ different from bin~$B_{u_i,w_j}$, the blocking dimensions with respect to edge~$\{u_i,w_j\}$ form a conflicting dimension set with respect to~$B$ of size~$Z$.
By \Cref{claim-bound-forbidden-move} we now have a contradiction and thus for each blocking vector~$\bvect(u_i,w_j)$ we have~$\bvect(u_i,w_j)\notin D_{\rm flip}(f,f')$.
\end{claimproof}
 
\emph{Step 3: Movement of the edge vectors in~$D_{\rm flip}(f,f')$.}
With the help of \Cref{claim-bound-forbidden-move} we verify the following.

\begin{claim}
\label{claim-movements-edge-vectors}
For each edge vector~$\vect(u_i,w_j)\in D_{\rm flip}(f,f')$, we have~$f'(\vect(u_i,w_j))= B^*$.
\end{claim}
\begin{claimproof}
Let~$\vect(u_i,w_j)\in D_{\rm flip}(f,f')$.
Recall that~$f(\vect(u_i,w_j))=B_{u_i,w_j}$.
Observe that~$\vect(u_i,w_j)$ has value~1 in each dummy edge dimension with respect to edge~$\{u_i,w_j\}$.
Also, note that in the initial assignment~$f$, no bin except the target bin~$B^*$ is empty in any dummy edge dimension with respect to edge~$\{u_i,w_j\}$.
More precisely, each vector except vector~$\vect(u_i,w_j)$ having value~1 in at least one dummy edge dimension with respect to edge~$\{u_i,w_j\}$ is a dummy vector.
Thus, each such vector has exactly one~$1$ in one dummy edge dimension with respect to edge~$\{u_i,w_j\}$.
Consequently, for each bin~$B$ different from bin~$B_{u_i,w_j}$ and~$B^*$, the dummy edge dimensions with respect to edge~$\{u_i,w_j\}$ form a conflicting dimension set with respect to~$B$ of size~$Z$.
By \Cref{claim-bound-forbidden-move} we now obtain that for each edge vector~$\vect(u_i,w_j)\in D_{\rm flip}(f,f')$, we have~$f'(\vect(u_i,w_j))= B^*$.
\end{claimproof}

\emph{Step 4: Movement of the vertex vectors in~$D_{\rm flip}(f,f')$.}
Again with the help of \Cref{claim-bound-forbidden-move} we verify the following.

\begin{claim}
\label{claim-movements-vertex-vectors}
For each vertex vector~$\vect(v_i)\in D_{\rm flip}(f,f')$, we have~$f'(\vect(v_i))= B_{v_i}$.
\end{claim}
\begin{claimproof}
Let~$\vect(v_i)\in D_{\rm flip}(f,f')$.
recall that~$f(\vect(v_i))=B^*$.
Observe that~$\vect(v_i)$ has value~1 in each dummy vertex dimension with respect to~$v_i$.
Also, note that in the initial assignment~$f$, no bin except the vertex bin~$B_{v_i}$ is empty in any dummy vertex dimension with respect to~$v_i$.
More precisely, each vector except vector~$\vect(u_i,w_j)$ having value~1 in at least one dummy vertex dimension with respect to~$v_i$ is a dummy vector.
Thus, each such vector has exactly one~$1$ in one dummy vertex dimension with respect to~$v_i$.
Consequently, for each bin~$B$ different from bins~$B^*$ and~$B_{v_i}$, the vertex dimensions with respect to~$v_i$ form a conflicting dimension set with respect to~$B$ of size~$Z$.
By \Cref{claim-bound-forbidden-move} we now obtain that vertex vector~$\vect(v_i)\in D_{\rm flip}(f,f')$ we have~$f'(\vect(v_i))= B_{v_i}$.
\end{claimproof}

\emph{Step 5: No dummy vector is contained in~$D_{\rm flip}(f,f')$.}

\begin{claim}
\label{claim-no-dummy-vector-is-moved}
For each dummy vector~$\vect$, we have~$\vect\notin D_{\rm flip}(f,f')$.
\end{claim}
\begin{claimproof}
By \Cref{claim-do-not-move-blocking} we know that~$D_{\rm flip}(f,f')$ contains no blocking vector, by \Cref{claim-movements-edge-vectors}, we know that for each edge vector~$\vect(u_i,w_j)\in D_{\rm flip}(f,f')$ we have~$f'(u_i,w_j)=B^*$, and by \Cref{claim-movements-vertex-vectors}, we know that each vertex vector~$\vect(v_i)\in D_{\rm flip}(f,f')$ we have~$f'(v_i)=B_{v_i}$.
Thus, if we consider all moves of non-dummy vectors separately, that is, moves of blocking vectors, edge vectors, and vertex vectors, in the improving flip~$D_{\rm flip}(f,f')$, according to \Cref{obs-bin-packing-initial-assignment} dummy vectors only 'overlap' with vertex vectors (which are contained in~$D_{\rm flip}(f,f')$) in vertex bins.
Here, overlap means that the dummy vector is in the same bin as the vertex vectors and both have value 1 in one dimension.
More precisely, for each vertex vector~$\vect(v_i)$ contained in the improving $k$-flip, this vertex vector~$\vect(v_i)$ overlaps in the vertex bin~$B_{v_i}$ with dummy vectors in the important profit dimensions~$d_{v_i,j}^y$ for each~$j\in[1,\ell]\setminus\{i\}$ and~$y\in[1,z/2-1]$ (this are the only dimensions where an overlap occurs).
According to \Cref{obs-bin-packing-initial-assignment-empty} of \Cref{obs-bin-packing-initial-assignment} in the initial assignment~$f$,
for each color~$i$ and each~$v_i\in V_i$ in vector bin~$B_{v_i}$ the important dimensions~$d_{v_i,j}^y$ for each~$j\in[1,\ell]\setminus\{i\}$ and each~$y\in[z/2,z]$ are empty.
In other words, no bin is empty with respect to any dimension~$d_{v_i,j}^y$ for each~$j\in[1,\ell]\setminus\{i\}$ and~$y\in[1,z/2-1]$.

Now we exploit the fact that~$f'$ is improving over~$f$ such that~$d_{\rm flip}(f,f')$ is smallest:
Moving any dummy-vector having value~1 in any such dimension~$d_{v_i,j}^y$ cannot reduce the overload, and thus~$f'$ contains no such dummy vector.
Furthermore, in no other dimension there is an overload which occurs due to a dummy vector and some other vector.
Consequently, $D_{\rm flip}(f,f')$ contains no dummy vector.
\end{claimproof}

\emph{Step 6: Implications of edge vector~$\vect(u_i,w_j)\in D_{\rm flip}(f,f')$.}
Next, we show the following.

\begin{claim}
\label{claim-edge-moved-only-if-both-vertices-moved}
For each edge vector~$\vect(u_i,w_j)\in D_{\rm flip}(f,f')$ we also have that the two vertex vectors~$\vect(u_i)$ and~$\vect(w_j)$ are contained in~$D_{\rm flip}(f,f')$.
\end{claim}
\begin{claimproof}
Without loss of generality, assume that~$\vect(u_i)\notin D_{\rm flip}(f,f')$.
Thus, $f(\vect(u_i))=B^*=f'(\vect(u_i))$.
Now, observe that in the assignment~$f'$, vertex vector~$\vect(u_i)$ and edge vector~$\vect(u_i,w_j)$ overlap in all important dimensions with respect to~$(u_i,j)$.
These are $z$~dimensions in total.
Furthermore, observe that in the initial assignment~$f$, the edge vector~$\vect(u_i,w_j)$ and the blocking vector~$\bvect(u_i,w_j)$ overlap in all important profit dimensions with respect to~$(u_i,j)$ and in all important profit dimensions with respect to~$(w_j,i)$.
Also, these are $z$~dimensions in total.
Recall that according to \Cref{claim-no-dummy-vector-is-moved}, it is safe to assume that~$D_{\rm flip}(f,f')$ does not move any dummy vectors.
Hence, by setting~$f^*(\vect(u_i,w_j))\coloneqq B_{u_i,w_j}=f(\vect(u_i,w_j))$ and~$f^*(\vect)\coloneqq f'(\vect)$ for each other vector~$\vect$, we obtain another solution which has the same overload as~$f'$ but~$d_{\rm flip}(f,f^*)=d_{\rm flip}(f,f')-1$.
This contradicts the assumption that~$d_{\rm flip}(f,f')$ is smallest among all improving flips and thus the statement is verified.
\end{claimproof}

\emph{Step 7: Putting everything together.}
Let~$k'\le k$ be the number of moved vectors.
\Cref{claim-no-dummy-vector-is-moved,claim-edge-moved-only-if-both-vertices-moved} imply that~$k'=\alpha+\beta$ where~$\alpha$ is the number of moved vertex vectors (according to \Cref{claim-movements-vertex-vectors} each of them is moved into its corresponding vertex bin), and~$\beta\le\binom{\alpha}{2}$ is the number of moved edge vectors (according to \Cref{claim-movements-edge-vectors} they are moved into the target bin~$B^*$).
We finally verify that~$\alpha=\ell$ and that~$\beta=\binom{\alpha}{2}=\binom{\ell}{2}$ implying that the vectors in~$D_{\rm flip}(f,f')$ correspond to  the vertices and edges of a multicolored clique in~$G$.

We show this statement in three steps: 
First, we show that~$\alpha\le\ell$, second we show that~$\beta=\binom{\alpha}{2}$, and third we verify that~$\alpha=\ell$.
In all three steps we compute the change in the overload as follows:
First, we calculate the increase of the overload by only considering the flips of the vertex vectors in~$D_{\rm flip}(f,f')$.
Then, all important edge dimensions where at least one edge vector in~$D_{\rm flip}(f,f')$ has value~1 are empty in~$B^*$.
Second, we calculate the decrease of the overload by only considering the flips of the edge vectors in~$D_{\rm flip}(f,f')$.
Recall that moving a vertex vector~$\vec(v_i)$ from the target bin~$B^*$ into its corresponding vertex bin~$B_{v_i}$ causes an overload of~$(\ell-1)\cdot (z/2-1)$.
Moreover, recall that moving an edge vector~$\vec(u_i,w_j)$ out of its corresponding edge bin~$B_{u_i,w_j}$ reduces the overload by~$z$.

\emph{Step 7.1: At most $\ell$~vertex vectors are contained in~$D_{\rm flip}(f,f')$:}
Assume towards a contradiction that at least~$\ell+1$ vertex vectors move.
Hence, the flip of the vertex vectors in~$D_{\rm flip}(f,f')$ increases the overload by at least~$(\ell+1)\cdot (\ell-1)\cdot(z/2-1)=(\ell^2-1)\cdot(z/2-1)$.
Furthermore, the flip of the edge vectors in~$D_{\rm flip}(f,f')$ decreases the overload by at most~$k-\ell-1=(\binom{\ell}{2}-1)\cdot z$.
Now, observe that~$f'$ cannot be improving upon~$f$:
\begin{align*}
 && (\ell^2-1)\cdot(z/2-1) &&\ge&& \left(\binom{\ell}{2}-1\right)\cdot z \\
\Leftrightarrow&& -\ell^2-z/2+1 &&\ge&& -\ell\cdot z/2-z  \\
\Leftrightarrow&& (\ell+1)\cdot z/2 - \ell^2+1 &&\ge&& 0 
\end{align*}
The last inequality is true since~$z=4\cdot\ell^2$.
Thus, the overload of~$f'$ is larger than the overload of~$f$, a contradiction to the fact that~$f'$ is improving upon~$f$.
Consequently, at most~$\ell$ vertex vectors are contained in~$D_{\rm flip}(f,f')$.

\emph{Step 7.2: Exactly $\binom{\alpha}{2}$~edge vectors are contained in~$D_{\rm flip}(f,f')$:}
Assume towards a contradiction that~$\beta\le (\binom{\alpha}{2}-1)$.
Hence,  the flip of the vertex vectors in~$D_{\rm flip}(f,f')$ increases the overload by exactly~$\alpha\cdot (\ell-1)\cdot(z/2-1)$. Furthermore, the flip of the edge vectors in~$D_{\rm flip}(f,f')$ decreases the overload by at most~$(\binom{\alpha}{2}-1)\cdot z$.
Now, observe that~$f'$ cannot be improving upon~$f$:
\begin{align*}
 && \alpha\cdot (\ell-1)\cdot(z/2-1) &&\ge&& \left(\binom{\alpha}{2}-1\right)\cdot z  \\
\Leftrightarrow&& \alpha\cdot\ell\cdot z/2 - \alpha^2\cdot z/2 -\alpha\cdot\ell +\alpha +z &&\ge&& 0  
\end{align*}
Since~$\alpha\le\ell$, we obtain~$\alpha\cdot\ell\cdot z/2 - \alpha^2\cdot z/2 \ge 0$ and since~$z=4\cdot\ell^2$, we obtain~$-\alpha\cdot\ell +\alpha +z\ge 0$, implying the inequality.
Thus, the overload of~$f'$ is larger than the overload of~$f$, a contradiction to the fact that~$f'$ is improving upon~$f$.
Consequently, exactly $\binom{\alpha}{2}$~edge vectors are contained in~$D_{\rm flip}(f,f')$.

\emph{Step 7.3: Exactly $\ell$~vertex vectors are contained in~$D_{\rm flip}(f,f')$:}
Hence,  the flip of the vertex vectors in~$D_{\rm flip}(f,f')$ increases the overload by exactly~$\alpha\cdot (\ell-1)\cdot(z/2-1)$. Furthermore, the flip of the edge vectors in~$D_{\rm flip}(f,f')$ decreases the overload by exactly~$\binom{\alpha}{2}\cdot z$.
We now show, that~$f'$ can only by improving upon~$f$ if~$\alpha=\ell$.
In contrast to Steps~7.1 and~7.2 we now verify that the decrease of the overload due to the edge vectors is larger than the increase due to the vertex vectors.
\begin{align*}
 && \binom{\alpha}{2}\cdot z &&>&& \alpha\cdot (\ell-1)\cdot(z/2-1) \\
 \Leftrightarrow&& \alpha(\alpha-\ell)z/2+\alpha\ell-\alpha &&>&& 0
 \\
 \Leftrightarrow&& \alpha -\ell + 2\cdot(\ell-1)/z &&>&& 0 
\end{align*}
Observe that since~$z=4\cdot\ell^2$, we have~$0<2\cdot(\ell-1)/z<1$.
This implies that~$\alpha\ge\ell$.
Recall that in Step~7.1 we showed that~$\alpha\le\ell$.
Thus,~$\alpha=\ell$. 
With Step~7.2, we conclude that exactly $\binom{\ell}{2}$~edge vectors are contained in~$D_{\rm flip}(f,f')$, implying that~$G$ has a multicolored clique.
\end{proof}

\section{Discussion}

In \Cref{Section: Framework} we considered parameterized local search for a very generic partitioning problem \GBP which generalized many important optimization problems like \VC, \CE, or \MK. 
For the local search version of this problem, in \Cref{Theorem: Framework Algorithms} we derived algorithms that run in $\tau^k \cdot 2^{\Oh(k)}\cdot |I|^{\OO(1)}$ and $k^{\OO(\tau)}\cdot |I|^{\OO(1)}$~time, respectively. 
Here, $k$ is the search radius and $\tau$ is the number of \emph{types}.
In \Cref{Section: Framework Application} these algorithms were then used as a framework to obtain efficient parameterized local search algorithm for many classic problems by reducing many classical optimization problems to~\GBP.
For example, in \Cref{Theorem: Algo Cluster Editing}, we provided an algorithm with running time~$\nd^k \cdot 2^{\Oh(k)}\cdot |I|^{\Oh(1)}$ for the local search variant of \CE, and in \Cref{Theorem: Algo Vector Bin Packing}, we provided an algorithm with running time~$\tau^k \cdot 2^{\Oh(k)}\cdot |I|^{\Oh(1)}$ for the local search variant of \Bin where~$\tau$ is the number of distinct vectors.
Moreover, for each considered problem where W[1]-hardness for~$k$ and an ETH-lower bound of~$|I|^{o(k)}$ was not known before, we provided such a result.
Consequently, for each considered problem our algorithms with running time~$\tau^k \cdot 2^{\Oh(k)}\cdot |I|^{\OO(1)}$ are essentially optimal.

There are several ways of extending our work:
In all of our applications of the framework we extensively used the expressive power of the IBEs and the flexibility of the bins.
So far, we have not exploited the power of the types. 
For future work it is thus interesting to find examples where the number~$\tau$ of types is smaller than the neighborhood diversity or the number of distinct vectors.
Also, it is interesting to study the parameterized complexity of the non-local search versions of the considered problems parameterized by~$\tau$ alone. For example, \textsc{Bin Packing} admits an FPT-algorithm when parameterized by~$\tau$ alone~\cite{K+21}. It appears to be reasonable that the algorithm can be modified to obtain fixed-parameter tractability for \textsc{Vector Bin Packing} parameterized by~$\tau$ as well. Furthermore, \VC admits an FPT-algorithm when parameterized by~$\nd$~\cite{koutecky2013solving} and~\MC admits an FPT-algorithm when parameterized by~$\nd+c$~\cite{GKK22}.
However, it is not known whether such algorithms are possible for all (classic) variants of our study; for example for \NSW such an algorithm is not known.

\bibliographystyle{plainurl} 
\bibliography{refs}

\end{document}